\documentclass[11pt,a4paper]{article}

\usepackage[T1]{fontenc}
\usepackage[utf8]{inputenc}
\usepackage{amsmath,amssymb,amsthm,amsfonts,mathtools,dsfont}
\usepackage[margin=1in]{geometry} 
\usepackage{comment}
\usepackage{mathrsfs}
\usepackage{hyperref}

\usepackage{booktabs}
\usepackage{chngcntr}		
\counterwithin{table}{section}
\counterwithin{figure}{section}
\usepackage{tikz}
\usetikzlibrary{arrows}
\usepackage[round, compress]{natbib}

\theoremstyle{plain}
\newtheorem{theorem}{Theorem}[section]
\newtheorem{proposition}[theorem]{Proposition}
\newtheorem{lemma}[theorem]{Lemma}

\theoremstyle{definition}
\newtheorem{definition}[theorem]{Definition}
\newtheorem{remark}[theorem]{Remark}
\newtheorem{example}[theorem]{Example}
\newtheorem{assumption}[theorem]{Assumption}

\theoremstyle{remark}
\numberwithin{equation}{section}

\newcommand{\RR}{\mathbb{R}}

\newcommand{\NN}{\mathbb{N}}

\newcommand{\diag}{\mathrm{diag}}

\newcommand{\ol}{\overline}

\newcommand{\cA}{\mathcal{A}}

\definecolor{pipurple}{RGB}{128,0,128}
\definecolor{vargreen}{RGB}{0,150,0}
\definecolor{varyellow}{RGB}{180,120,0}

\newcommand{\ones}{\mathbf{1}}
\newcommand{\diff}{\mathrm{d}}
\newcommand{\dd}{\,\mathrm{d}}
\renewcommand{\epsilon}{\varepsilon}
\newcommand{\id}{\mathrm{id}}

\renewenvironment{thebibliography}[1]{%
\begin{oldthebibliography}{#1}%
\setlength{\baselineskip}{.8em}
\linespread{.8}
\small 
\setlength{\parskip}{0ex}%
\setlength{\itemsep}{.1em}%
}%
{%
\end{oldthebibliography}%
}

\usepackage{color}

\usepackage{comment}

\begin{document}

\title{Liquidity Provision with Adverse Selection and Inventory Costs}

\author{
Martin Herdegen\thanks{University of Warwick, Department of Statistics, \url{m.herdegen@warwick.ac.uk}.}
\and
Johannes Muhle-Karbe\thanks{Imperial College London, Department of Mathematics, email \url{j.muhle-karbe@imperial.ac.uk}.  Research supported by the CFM-Imperial Institute of Quantitative Finance.}
\and
Florian Stebegg\thanks{Columbia University, Department of Statistics, email \url{florian.stebegg@columbia.edu}.}
}

\date{July 26, 2021}

\maketitle

\begin{abstract}
We study one-shot Nash competition between an arbitrary number of identical dealers that compete for the order flow of a client. The client trades either because of proprietary information, exposure to idiosyncratic risk, or a mix of both trading motives. When quoting their price schedules, the dealers do not know the client's type but only its distribution, and in turn choose their price quotes to mitigate between adverse selection and inventory costs. Under essentially minimal conditions, we show that a unique symmetric Nash equilibrium exists and can be characterized by the solution of a nonlinear ODE.
\end{abstract}




\section{Introduction}

Trades in financial markets are typically executed either to profit from superior information or for idiosyncratic ``liquidity reasons'', e.g., to offload certain risks by hedging. 
 As succinctly summarized by~\cite{treynor.71}, ``the essence of market making, viewed as a business, is that in order for the market maker to survive and prosper, his gains from liquidity-motivated transactors must exceed his losses to information motivated transactions.'' Designated market makers have mostly been replaced by either centralized limit-order books or OTC markets comprised of discretionary dealers. Nevertheless, this basic tradeoff between exposure to adverse selection through counterparties with superior information and rents earned by servicing liquidity trades remains at the heart of liquidity provision.  

The other main risk that liquidity providers are exposed to is ``inventory risk'', incurred from ``uncertainty about the return on their inventory but also from the uncertainty about when future transactions will occur (which affects how long they must bear return uncertainty)'' \citep{ho.stoll.81}. 

This paper studies several strategic liquidity providers (``dealers''), who compete for the business of a liquidity-taking ``client'' by quoting price schedules, which indicate at what prices the dealers stand ready to fill orders of different sizes. As in the seminal works of~\cite*{biais.al.00,biais.al.13,back.baruch.13}, we assume that clients trade either because they have private information about the payoff of the risky asset, for liquidity reasons (because they are exposed to some idiosyncratic risk), or due to a mix of both trading motives. 

When quoting their price schedules, the dealers do not know the clients' type, but can only mitigate their adverse selection risk based on its distribution. In addition, as in \cite*{bielagk.al.19}, the dealers also take into account inventory risk through a quadratic cost on their post-trade positions.\footnote{The case of zero inventory costs studied in \cite{biais.al.00,biais.al.13,back.baruch.13} is natural in markets with very short holding times such as equities or currencies; inventory costs become increasingly important in markets with longer holding times such as (exotic) options. \cite{bielagk.al.19} incorporate convex holding costs into a model with a single dealer and a dark pool; we add such inventory costs to a Nash competition between several strategic dealers as in~\cite{biais.al.00}.}

In this setting, we show that there is a unique symmetric Nash equilibrium, where the dealers' optimal price schedules are characterized by a nonlinear ODE. The corresponding equilibrium prices naturally exhibit a bid-ask spread between the best buying and selling prices; as a result, only clients with sufficiently strong trading motives end up engaging with the dealers. Furthermore, we show that convexity of equilibrium price schedules for several competing dealers is endogenous, and does not have to be assumed a priori as in \cite{glosten.89,biais.al.00,biais.al.13,back.baruch.13}. While markets organized as limit-order books do not allow dealers to offer discounts for large quantities, justifying this assumption, we are therefore able to show that such discounts are also not optimal in other dealer-based markets. In contrast, quantity discounts can be optimal for markets dominated by a monopolist dealer.

We establish existence and uniqueness of a symmetric Nash equilibrium under very weak sufficient conditions on the distribution of clients. To wit, we consider general log-concave distributions for the client-type supported on the whole real line. In concrete examples with Gaussian or Exponential distributions, the characteristic ODE can be analyzed directly and optimality can in turn be verified by a direct verification argument~\citep{back.baruch.13}. We show that a unique solution of the ODE exists in general, despite the lack of natural boundary conditions, complementing results of~\cite{biais.al.00,biais.al.13} for type distributions with compact support. Our proof of existence is constructive in that it puts forward explicit upper and lower solutions. Using these as boundary conditions, standard ODE solvers for uniformly Lipschitz equations directly yield very tight upper and lower bounds of the optimal price schedule on any finite interval.

As pointed out by~\cite{back.baruch.13,biais.al.13}, a well-behaved solution of the characteristic ODE (derived from the dealers' first-order conditions) is generally \emph{not} sufficient to produce a Nash equilibrium. Instead, for risk-neutral dealers, sufficiently strong adverse selection by informed clients is required to guarantee that the solution of the ODE indeed leads to an equilibrium. We show that the dealers' inventory costs (relative to the clients') serve as a partial substitute for this, in that the solution of the ODE leads to an equilibrium \emph{if} a combination of adverse selection and inventory costs is sufficiently large. Sharp conditions are expressed in terms of the solution of the characteristic ODE; sufficient conditions in terms of primitives of the model can be derived from explicit upper and lower solutions. In the special case of Gaussian client-type distributions and risk-neutral dealers, this recovers the adverse-selection bounds of~\cite{back.baruch.13}.

From a mathematical perspective, the main challenges are that we do not work on compact type intervals as in \cite{biais.al.00,biais.al.13}, and that we do not make any convexity assumptions, neither for the clients' problem as in \cite{back.baruch.13} nor the dealer's problem as in \cite{biais.al.00}.  Also, we cannot rely on special properties of the Gaussian distribution as in \cite{back.baruch.13}.
	
For the client's problem, the lack of compactness implies that existence and continuity properties of the optimizer do not follow from Berge's theorem. Since -- unlike the extant literature -- we do not assume convexity a priori (but prove it for the oligopolistic case), establishing sufficient and necessary conditions for price schedules to be admissible is rather delicate. As a remedy, we combine the precise structure of the aggregated goal function (which is linear in the type variable of the client) with the a priori $C^2$-regularity of the goal function in the optimisation variable, and the necessary and sufficient first-order conditions of interior optimizers to gradually derive more and more structure for the optimizer as a function of the type variable.

For the dealers' problem, the lack of compactness implies that the candidate ODE in the oligopolistic case has no boundary conditions. To obtain existence and to apply numerical schemes, we therefore need to construct explicit upper and lower solutions with appropriate properties. Combining these with a taylor-made Gr\"onwall estimate in turn also allows to prove uniqueness.  The lack of convexity for the dealer's goal functional implies that to verify optimality (both in the monopolistic and the oligopolistic case), we cannot rely on second-order conditions but rather have to work with first-order conditions which are more general but delicate to establish, in particular, since the optimal price schedules exhibit a nontrivial bid-ask spread.

This article is organized as follows. In Section~\ref{s:client}, we study the liquidity-taking clients' problem, taking the price schedules quoted by the dealers as given. With the clients' optimal demand function at hand, we then turn to the dealers' problem in Section~\ref{s:nash}. As a benchmark, we first study in Section~\ref{ss:mon} the (most tractable) case of a single monopolistic dealer. Then, in Section~\ref{ss:oli}, we analyze the symmetric Nash competition between an arbitrary number of identical dealers. For better readability, all proofs are collected in the appendices.

\section{The Client's Problem}\label{s:client}

We consider $K \geq 1$ symmetric dealers who compete for the orders of a client in a Stackelberg-type game: the dealers ``lead'' by quoting price schedules that describe at what prices they are willing to trade various quantities. The client ``follows'' by choosing her optimal trade sizes.

As is customary for Stackelberg equilibria, we start by analyzing the optimal response function of the ``follower''. In the present context, this means that we first focus on the client and study her optimal trade sizes given the price schedules quoted by the dealers. This analysis needs to be carried out for heterogenous price schedules, even though we focus on identical dealers who quote the same price schedules in equilibrium. The reason is the Nash competition between the dealers that we will analyze in Section~\ref{s:nash}: for a Nash equilibrium, one needs to verify that unilateral deviations are suboptimal. Accordingly, the response function of the client is required for heterogeneous off-equilibrium price schedules, even if the equilibrium itself is symmetric.

For the client's problem, the prices quoted by the dealers are fixed. To wit, dealer $k \in \{1,\ldots,K\}$ quotes a \emph{price schedule} $P_k: \RR \to \RR$, i.e., a continuous function that satisfies $P_k(0)=0$ and is twice continuously differentiable on $\RR \setminus \{0\}$.\footnote{This means that prices are a smooth function of trade sizes except for a potential ``bid-ask spread'' between buying and selling prices. The equilibrium we construct in Section~\ref{s:nash} will be of this form.} This means that $n$ shares of the risky asset can be bought from dealer $k$ for $P_k(n)$ units of cash.  The client has inventory costs $\gamma_c>0$, an initial position $M=m$ in the risky asset, and receives a \emph{signal} $S=s$ about its payoff $V$:
$$V=S+\varepsilon.$$
Here, the mean-zero error term $\varepsilon$, the signal $S$, and the client's initial position $M$ are independent. The client observes the realizations $s$ and $m$ of her signal $S$ and inventory $M$; in contrast, the dealers only know the distributions of these random variables. Setting 
$$\mathbf{P}=(P_1,\ldots,P_K),$$
the client then chooses her trades 
$$\mathbf{n}=(n_1,\ldots,n_K)$$ 
to maximize her expected profits penalized for quadratic inventory costs $\gamma_c>0$ (if the error term $\varepsilon$ is normally distributed as in~\cite{biais.al.00}, then one could equivalently assume that the client maximizes her expected exponential utility):
\begin{equation}\label{eq:client0}
\bar{J}_c^{\bf P}(s,m;\mathbf{n}):=s\left(m + \sum_{k =1}^K n_k\right) - \sum_{k =1}^K P_k(n_k)- \frac{\gamma_c}{2} \left(m+ \sum_{k =1}^K n_k\right)^2.
\end{equation}
Here, the first term is the expected payoff of the client's post-trade position, conditional on her information set $M=m$ and $S=s$. The second term collects the cash payments to the dealers, and the last term describes the inventory cost of the client's post-trade position. 

We make the following standing assumption on the distribution of type variables;  the standard example are normal distributions as in, e.g., \cite{glosten.89}, but many other common distributions such as two-sided exponential or gamma distributions also fall into this framework. In the appendix, we prove our results under weaker but less intuitive conditions; these also allow to cover other distributions, e.g., of Pareto type.\footnote{Distributions with compact support are studied in \cite{biais.al.00} or \cite{cetin.21}, for example; discrete types are analyzed by \cite{attar.al.19}.}

\begin{assumption}\label{ass:type}
The distributions of the type variables $M$ and $S$ have positive densities $f_m$ and $f_S$  on $\RR$ that are \emph{log-concave}, i.e., $\log(f_M)$ and $\log(f_S)$ are concave functions
\end{assumption}

While the client's problem apparently depends on the two type variables $M=m$ and $S=s$, their impact on the client's problem can in fact be summarized by a single ``effective'' type as already observed by~\cite{glosten.89}:
$$
Y:=S-\gamma_c M.
$$
Indeed, for each realization $Y=y=s-\gamma_c m$, maximizing~\eqref{eq:client0} over $\mathbf{n}=(n_1,\ldots,n_K)$ is equivalent to maximizing the ``normalized'' goal functional
\begin{align}
J_c^{\bf P}(y; {\bf n}) &:= \bar{J}^{\bf P}_c(s,m;\mathbf{n})-sm-\frac{\gamma_c}{2}m^2  =y\sum_{k =1}^K n_k-\sum_{k =1}^K P_k(n_k) - \frac{\gamma_c}{2} \left(\sum_{k =1}^K n_k\right)^2.\label{eq:client}
\end{align}
We denote the density of the  effective type variable $Y$ by $f$. Assumption~\ref{ass:type} implies that $f$ is log-concave and continuously differentiable with bounded derivative; see Proposition \ref{prop:log concave:convolution}(a). The corresponding cumulative distribution function and survival function are denoted by
$$
F(x) := \int_{-\infty}^x f(y) \dd y \quad \mbox{and} \quad \ol F(x) = 1 - F(x) = \int_x^\infty f(y) \dd y, \quad x \in \RR.
$$ 

\begin{remark}
The client's optimization only depends on her effective type $Y=S-\gamma_c M$. We nevertheless impose log-concavity on both $S$ and $M$, which implies log-concavity of $Y$, rather than assuming this directly as in \cite{biais.al.00}. This provides an intuitive condition in terms of the primitives of the model -- as pointed out by~\cite{miravete:02}, only assuming log-concavity of either $S$ or $M$ does not guarantee log-concavity of $Y$. Moreover, log-concavity of both signals and inventories already ensures that most other regularity conditions required for the subsequent analysis hold automatically. For example, by Proposition~\ref{prop:g:mon}, log-concavity of the primitives already implies the bounds on the sensitivity of the expected asset payoff conditional on the client's effective type assumed directly in~\cite[p.~807]{biais.al.00}.
\end{remark}

As a minimal requirement for unique symmetric Nash equilibria (with pure strategies), we impose that dealers only quote price schedules $P_k$ for which a symmetric quote $P_1,\ldots,P_K=P$ gives rise to a unique solution of the clients' problem~\eqref{eq:client} for all realizations $y$ of the type $Y$:

\begin{definition}\label{def:adm}
A price schedule $P$ is \emph{admissible (for $K$ dealers)} if -- for every type $y \in \mathbb{R}$ and given that all dealers quote the price schedule $P$ -- the client's goal function 
$$J_c^P(y;\mathbf{n}):=J^{(P,\ldots,P)}(y; \mathbf{n})$$ 
has a unique maximizer ${\bf n}^P(y) = (n^P(y), \ldots, n^P(y))$.\footnote{Note that if $K \geq 2$ and a unique maximizer ${\bf n}^P = (n_1^P(y), \ldots, n_K^P(y))$ exists, then it follows from symmetry of the function $J^P_c(y;n_1, \ldots, n_K)$ in $n_1, \ldots, n_K$ that $n_1^P(y)= \ldots = n_K^P(y)$.}
\end{definition}

The following lemma collects the properties of the client's optimal trade sizes $n^P$ for identical, admissible price schedules quoted by all dealers. 

\begin{lemma}
	\label{lem:adm schedule}
Let $P: \RR \to \RR$ be an admissible price schedule for $K \geq 1$ dealers and denote by 
$$I^P := \left\{y \in \RR: n^P(y) \neq 0\right\}$$
the set of types for which trading is optimal. Then:
\begin{enumerate}
\item The function $y \mapsto n^P(y)$ is continuous;
\item There are $-\infty < a^P \leq b^P <\infty $ such that $I^P = (-\infty, a^P) \cup (b^P, +\infty)$;
\item The function $n^P$ is increasing on $I^P$ with $n^P(I^P) = \RR \setminus \{0\}$ and satisfies
\begin{align*}
\lim_{y \to -\infty} n^P(y) = -\infty,  \quad \lim_{y \uparrow a^P} n^P(y) = 0 =  \lim_{y \downarrow b^P} n^P(y), \quad \text{and} \quad \lim_{y \to +\infty} n^P(y) = \infty.
\end{align*}
\end{enumerate}
\end{lemma}

Properties (b) and (c) in Lemma~\ref{lem:adm schedule} show that when facing identical admissible price schedules, client types are divided into three categories. To wit, clients optimally sell to the dealers if their type is sufficiently small $y$ (i.e., they either receive signals indicating sufficiently unfavorable payoffs, or hold large initial inventories that need to be reduced to limit  inventory costs). Conversely, clients with large type $y$ purchase risky shares since they expect favorable payoffs or because they start with a negative inventory that they wish to reduce. Clients with an intermediate type $y \in [a^P,b^P]$ in turn prefer not to trade with the dealers. Moreover, Property (c) shows that there is no limit on what the client will buy or sell if their type is sufficiently small or large. Property (a) asserts that the dependence of the trade size on the type is smooth. 

With competition between two or more dealers, admissibility in the sense of Definition~\ref{def:adm} is equivalent to convexity of the price schedules, which is often assumed a priori~\citep{biais.al.00,biais.al.13,back.baruch.13}. In contrast, admissible price schedules for a monopolistic dealer need not be strictly convex and can in fact be strictly concave (compare~\cite[Section~4.2]{biais.al.00}):

\begin{lemma}\label{lem:schedules:hom}
	\begin{enumerate}
	\item A price schedule $P$ is admissible for $K \geq 2$ dealers if and only if it is strictly convex.
	\item	A price schedule $P$ is admissible for $K=1$ dealer if and only if $n \mapsto \frac{\gamma_c}{2} n^2 + P(n)$ is strictly convex on $\RR$ and satisfies $\lim_{n \to \pm\infty} \left(\gamma_c n +P'(n)\right) = \pm\infty$.
	\end{enumerate}
\end{lemma}

\begin{remark}
\begin{enumerate}
\item[(a)] By Lemma~\ref{lem:schedules:hom}, admissible price schedules for $K \geq 2$ dealers are also admissible for a monopolist. By contrast, admissibility for a monopolist does not necessarily imply admissibility for two or more dealers.
\item[(b)] For a monopolist dealer ($K = 1$), admissible price schedule can be concave, and more concave if the client is more inventory averse.
\end{enumerate}
\end{remark}

Admissibility of a price schedule $P$ in the sense of Definition~\ref{def:adm} requires that the client's problem has a unique solution for all types $Y=y$ given that all dealers post $P$. This is a natural minimal requirement for symmetric equilibria, but not sufficient to analyse Nash competition between the dealers. Indeed, if one of the dealers unilaterally deviates from a symmetric Nash equilibrium, then the client will face asymmetric price schedules. Whence, her optimization problem is no longer guaranteed to have a (unique) solution even if the price schedules in the equilibrium are admissible. A natural way out is to focus on deviations for which the asymmetric price schedules remain ``compatible'', in that the client's problem still has a solution at least for some type:

\begin{definition}
	Admissible price schedules $P_1, \ldots, P_K$ for $K \geq 2$ dealers are \emph{compatible} if the client's (normalized) problem $\max_{\mathbf{n}} J^\mathbf{P}(y;\mathbf{n})$ has a solution for some type $y \in \RR$. (Uniqueness then follows automatically since the strict convexity of the admissible price implies that the client's goal functional is strictly concave, cf.~the proof of Theorem~\ref{thm:schedules het}.)\end{definition}
	
Note that compatibility of the price schedules (and whence existence for some client types) can fail even if the price schedules are integrals of increasing marginal prices as in \cite{biais.al.00}, cf.~Example~\ref{ex:counter}.
	
For not necessarily identical but compatible price schedules, \emph{all} clients' problems are still well posed. The next result provides a simple criterion for compatibility in terms of limiting marginal prices and collects the properties of the corresponding optimal trade sizes:

\begin{theorem}
\label{thm:schedules het}
Let $P_1, \ldots, P_K$ be admissible price schedules for $K \geq 2$ dealers and set
\begin{equation}\label{eq:lr}
	\ell_{k} := \lim_{n \to -\infty} P'_k(n) \quad \mbox{and} \quad r_{k} := \lim_{n \to \infty} P'_k(n) \quad \mbox{for $k \in \{1, \ldots, K\}$.}
\end{equation}
Then $P_1, \ldots, P_K$ are compatible if and only if 
\begin{equation}\label{eq:ol:lr}
\max_{k \in \{1, \ldots K\}} \ell_k =: \ol \ell < \ol r := \min_{k \in \{1, \ldots K\}} r_k.
\end{equation}
 Moreover, in this case, setting ${\bf P} = (P_1, \ldots, P_k)$, we have the following properties:
\begin{enumerate}
\item For \emph{every} type $y \in \mathbb{R}$, the client's (normalized) problem $\max_{\mathbf{n}} J^\mathbf{P}(y;\mathbf{n})$ has a \emph{unique} solution ${\bf n}^{\bf P}(y) = (n_1^{\bf P}(y), \ldots, n_K^{\bf P}(y))$;
\item For each $k \in \{1, \ldots, K\}$, the function $n_k^{\bf P}: \RR \to \RR$ is continuous, nondecreasing, and increasing on $I^{\bf P}_k := \{y \in \RR: n_k^{\bf P}(y) \neq 0\}$;
\item For each $k \in \{1, \ldots, K\}$, $n_k^{\bf P}(I^{\bf P}_k) = \big((P'_k)^{-1}(\ol \ell),  (P'_k)^{-1}(\ol r)\big) \setminus \{0\}$,
where $(P'_k)^{-1}(m) := 0$ if $m \in [P'_k(0-), P'_k(0+)]$.
\end{enumerate}
\end{theorem}

As for the symmetric case, Theorem~\ref{thm:schedules het}(b) and (c) show that -- for each dealer -- clients with sufficiently negative type sell, clients with sufficiently positive type buy, and clients with intermediate types do not trade at all. However, unlike in the symmetric case, each of these three intervals (``buy'', ``no-trade'' and ``sell'') can now be empty for a given dealer. In particular, part (c) shows that there may be a limit on what the client will buy or sell from a given dealer irrespective of their type. Note, however, that compatibility ensures that at least one of the buy, no-trade and sell intervals is nonempty for each dealer, and there exists a at last one dealer where there is no limit on how much the client will buy if their type is large enough (if $r_k = \ol r$) and there is at last one dealer where there is no limit on how much the client will sell if their type is small enough (if $\ell_k = \ol \ell$).

\section{The Dealers' Problem}\label{s:nash}

\subsection{Monopolistic Case}\label{ss:mon}

As a benchmark, we first consider the most tractable case of a single monopolistic dealer, which has been studied by \cite{glosten.89} for risk-neutral preferences and Gaussian primitives. Generally, the monopolistic dealer chooses an admissible price schedule $P$ to maximise expected profits penalized for quadratic inventory costs:\footnote{Here, we use the convention that the integral equals minus infinity if its negative part is not integrable.}
\begin{align}
J_d(P) &:= E\left[ P(n^P(Y))- V n^P(Y)-\frac{\gamma_d}{2}n^P(Y)^2\right] \notag \\
&= \int_{-\infty}^\infty E\left[ P(n^P(Y))- V n^P(Y)-\frac{\gamma_d}{2}n^P(Y)^2\,\Big|\,Y=y\right]f(y) \dd y \notag\\
&= \int_{-\infty}^\infty \left(P(n^P(y))-E[V\,|\,Y = y] n^P(y)-\frac{\gamma_d}{2}n^P(y)^2\right)f(y)\dd y. \label{eq:mon}
\end{align}
The first term on the right-hand side of \eqref{eq:mon} are the cash payments received from the clients for selling the risky asset. The second term describes the corresponding payoff of the risky assets, and the third is a quadratic cost $\gamma_d \geq 0$ levied on the post-trade inventory. For example, in the context of exotic options, the inventory cost can be a multiple of the portfolio variance after hedging the option.

As the clients' optimal response function $n^P(\cdot)$ from Section~\ref{s:client} only depends on their effective type $Y=y$, the relevant statistic for the payoff of the risky asset is its corresponding conditional expectation $E[V\,|\,Y = y]$. For log-concave client types as in Assumption \ref{ass:type}, this conditional-mean function has the following properties:
	
	\begin{proposition}\label{prop:g:mon}
		The expected payoff of the risky asset conditional on the clients' type,
		\begin{equation}\label{eq:condexp}
		y \mapsto g(y) := E[V\,|\,Y = y]
		\end{equation}
		is continuously differentiable with $0 < g'(y) < 1$ for all $y \in \RR$.
	\end{proposition}
	
The properties $0<g'$ and $g'<1$ mean that the agents' type has some positive correlation with the asset payoff, but is not perfectly informative. This is assumed in \cite[p.~807]{biais.al.00} but in fact holds automatically if the densities of both signals and inventories are log concave.	

 \begin{example}\label{ex:normal}
 Suppose the client's inventory $M$ and signal $S$ are both normally distributed with means $\mu_M$ and $\mu_S$ and variances $\sigma^2_M$ and $\sigma^2_S$, respectively.  Then $Y=S-\gamma_c M$ is also normally distributed with mean $\mu_Y = \mu_S - \gamma_c \mu_M$ and variance $\sigma^2_Y = \sigma^2_S + \gamma_c^2 \sigma^2_M$. Moreover, the conditional mean is a linear function of the client's type:
 \begin{align}
 g(y) &=E[V\,|\,Y=y] =E[S\,|\,S-\gamma_c M=y] \notag\\
 &=E[S]+\frac{\mathrm{Cov}[S,S-\gamma_c M]}{\mathrm{Var}[S-\gamma_c M]}(y-E[S-\gamma_c M]) \notag\\
  &=(1-\beta)\mu_s +\beta \gamma_c \mu_M +\beta y \notag \\
 &=(1-\beta)\mu_y  + \gamma_c \mu_M +\beta y, \quad \mbox{for } \beta=\frac{\sigma_S^2}{\sigma_S^2+\gamma_c^2 \sigma_M^2} \in (0,1). \label{eq:projection}
 \end{align}
 \end{example}	

 \begin{example}\label{ex:same}
Linear functions $g(y)$ as in~Example~\eqref{ex:normal} arise more generally if the client's signal and inventory are both generated using iid copies of the same distribution. To wit, suppose that $S=\bar{S}+S^{(1)}+\ldots+S^{(k)}$ and $\tilde{M}:=-\gamma_c M=-\gamma_c M^{(1)}-\ldots-\gamma_c M^{(l)}=\tilde{M}^{(1)}+\ldots+\tilde{M}^{(l)}$, where $\bar{S} \in \mathbb{R}$ and all random variables are iid and have the same log-concave distribution. Then, by symmetry, the conditional mean is indeed affine linear in the type:
\begin{align*}
g(y) &=E\left[S\,|\,S-\gamma_c M=y\right] \\
&=\bar{S}+k E\left[S^{(1)}\,|\,S^{(1)}+\ldots+S^{(k)}+\tilde{M}^{(1)}+\ldots+\tilde{M}^{(l)}=y\right] =\bar{S}+\frac{k}{k+l}y.
\end{align*}
\end{example}

	For log-concave distributions of the client types as in Assumption~\ref{ass:type}, the cumulative distribution functions $F$ and survival function $\bar F$ of the effective type $Y$ are log-concave as well (cf.~Proposition \ref{prop:log:concave:basics}(a)). In particular, both $F/f$ and $-\bar F/f$ are nondecreasing. Denoting by $\id$ the identity function, this implies via Proposition~\ref{prop:g:mon} that the functions $F/f + \mathrm{id} - g$ and $-\ol F/f + \mathrm{id} - g$ are increasing and continuously differentiable with positive derivatives. In particular, their inverses are well defined, increasing and continuously differentiable. They in turn determine the monopolistic dealer's optimal price schedule in closed form:

	\begin{lemma}
		\label{lem:dealer:mon}
		Suppose Assumption \ref{ass:type} is satisfied. Then the monopolistic dealer's problem~\eqref{eq:mon} has a unique solution $P_*(n)=\int_0^n P_*'(x)\dd x$, where the optimal marginal prices are given by
		\begin{equation}
			\label{eq:lem:dealer:mon}
			P'_*(n) = \begin{cases}
				\left(\frac{F}{f} + \mathrm{id} - g\right)^{-1}\Big((\gamma_d  + \gamma_c)n\Big) - \gamma_c n, &\text{ if } n < 0, \\
				\left(-\frac{\ol F}{f} + \mathrm{id} - g\right)^{-1}\Big((\gamma_d + \gamma_c)n\Big) - \gamma_c n, &\text{ if } n > 0.
			\end{cases}
		\end{equation}
	\end{lemma}

\begin{remark}\label{rem:normal2}
Note that log-concavity of the client types is the only assumption needed for this result. In particular, no parameter constraints are required for Gaussian primitives, for example. If one additionally wants to guarantee that the optimal price schedule is convex (which is natural for a limit-order book, but not necessarily for a market dominated by a monopolist dealer) then additional conditions are needed. To wit, differentiation of~\eqref{eq:lem:dealer:mon} shows that the optimal price schedule for the monopolist is convex if and only if \begin{align}
	 (F/f)' - g' &\leq \frac{\gamma_d}{\gamma_c} \quad \text{on } \left(-\infty, 	\left(F/f+ \mathrm{id} - g\right)^{-1}(0)\right), 	\label{eq:rem:mon:est-}\\
(-\bar F/f)' - g' &\leq \frac{\gamma_d}{\gamma_c} \quad \text{on } \left(	\left(-\bar F/f + \mathrm{id} - g\right)^{-1}(0), \infty\right). \label{eq:rem:mon:est+}
\end{align}
For normally distributed inventories and signals as in Example \ref{ex:normal}, this can be translated into explicit parameter constraints. Suppose for simplicity as in \cite{back.baruch.13} that the dealer has no inventory costs ($\gamma_d = 0$). Then the optimal price schedule for the monopolist is convex if and only if 
\begin{equation}
\label{eq:rem:mon:beta est}
\tfrac{\gamma_c^2 \mu_M^2}{\sigma_S^2 + \gamma_c^2  \sigma_M^2} < \frac{\pi}{2} 
\quad  \text{and} \quad \beta \geq 1 + z_{\textrm{mon}}(\beta )\frac{\Phi}{\phi} \left( z_{\textrm{mon}}(\beta )\right),
\end{equation}
where $\Phi$ and $\phi$ denote the cumulative distribution function and the probability density function of the standard normal distribution, respectively, and 
\begin{equation*}
z_{\textrm{mon}}(\beta) = \tfrac{\gamma_c}{2 (1 -\beta)} \tfrac{|\mu_M|}{\sqrt{\sigma_S^2 + \gamma_c^2 \sigma_M^2}} -  \sqrt{\tfrac{\gamma^2_c \mu_M^2}{4 (1 -\beta)^2 (\sigma_S^2 + \gamma_c^2 \sigma_M^2)} + 1}.
\end{equation*}
In particular if $\mu_M = 0$, then $z_{\textrm{mon}}(\beta) = -1$ and \eqref{eq:rem:mon:beta est} specializes to the condition of~\cite{back.baruch.13}:
\begin{equation*}
	\beta \geq 1- \frac{\Phi}{\phi} (-1) \approx 0.3422.
\end{equation*}
This parameter constraint means that there is sufficient adverse selection to discourage the monopolist from offering quantity discounts for large trades. For inventories with non-zero mean, the second part of the corresponding condition~\eqref{eq:rem:mon:beta est} has the same interpretation. The first part additionally requires that client inventories are not too large on average. If they are, then large trades are likely enough to come from uninformed traders so that quantity discounts may be optimal. The crossover from convex optimal price schedules (i.e., increasing marginal prices $P'$) to locally concave ones is illustrated in Figure~\ref{fig:monopoly}. If the dealer is sufficiently inventory-averse relative to the client, then~\eqref{eq:rem:mon:est-}, \eqref{eq:rem:mon:est+} show that this effect disappears.
\end{remark}

\begin{figure}[htbp]
\begin{center}
\includegraphics[width=0.45\textwidth]{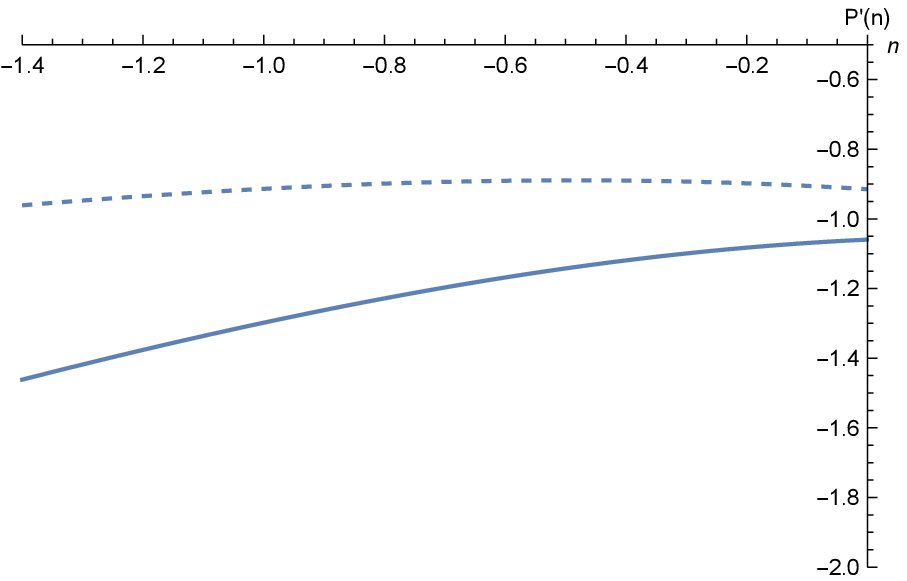}
\includegraphics[width=0.45\textwidth]{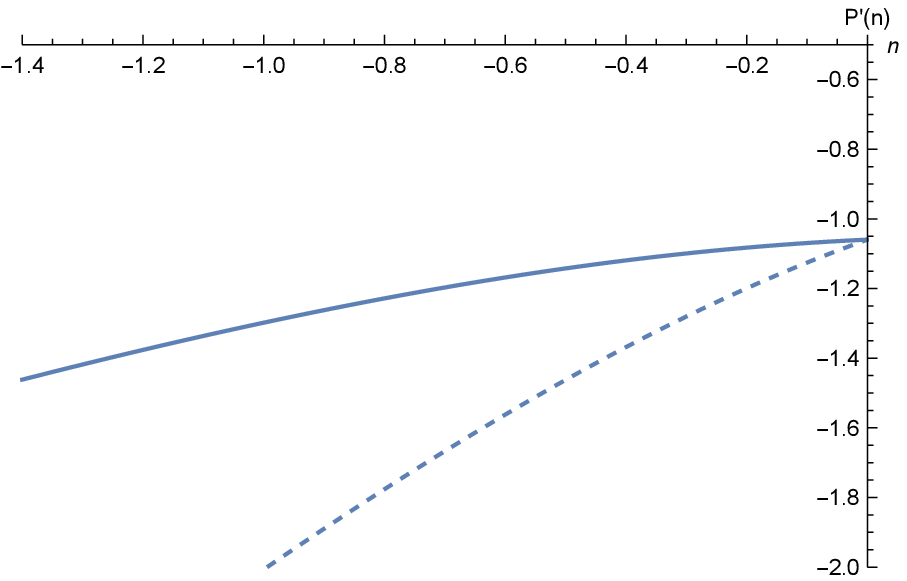}
\end{center}
\caption{Optimal marginal prices $P'(n)$ for monopolist dealers for normally distributed noise, client signals and inventories with $\beta=0.4$ (solid) and $\beta=0.25$ (dashed) and inventory costs $\gamma_c=1$, 
$\gamma_d=0$ (left panel). The right panel shows the marginal prices for $\beta=0.4$, $\gamma_c=1$ and $\gamma_d=0$ (solid), $\gamma_d=0.5$ (dashed).}\label{fig:monopoly} 
\end{figure}

The right panel in Figure~\ref{fig:monopoly} displays the dependence of the monopolist's optimal marginal prices on the corresponding inventory costs. The bid-ask spread is independent of the inventory costs here and~\eqref{eq:lem:dealer:mon} shows that this in fact holds in general. Whereas monopolist spreads remain invariant, larger inventory costs of the dealer lead to steeper marginal price curves and hence (more) convex price schedules. Differentiation of~\eqref{eq:lem:dealer:mon} shows that this also holds in general.

\subsection{Oligopolistic Case}\label{ss:oli}

We now turn to competition between several $K \geq 2$ identical, strategic dealers. To identify a symmetric Nash equilibrium, we suppose without loss of generality that the dealers $k \in \{2, \ldots, K\}$ post the same admissible price schedule $P$ and dealer $k=1$ then chooses an admissible price schedule $P_1$ such that ${\bf P} =  (P_1, P, \ldots, P)$ is compatible. The common price schedule $P$ in turn is a \emph{Nash equilibrium} if dealer $k=1$ has no incentive to deviate, in that the common price schedule is also optimal for her.

After fixing the price schedules of the other dealers, the goal functional of dealer 1 is 
$$
J_d^P(P_1) := \int_{-\infty}^\infty \left(P_1(n^{\bf P}_1(y))-g(y) n^{\bf P}_1(y)-\frac{\gamma_d}{2}n^{\bf P}_1(y)^2\right)f(y)\dd y.
$$
This is of the same form as the goal functional~\eqref{eq:mon} of the monopolistic dealer, except that the trade $n^{\bf P}_1(y)$ that the client conducts with dealer 1 now depends on the price schedules ${\bf P} =  (P_1, P, \ldots, P)$ quoted by \emph{all} dealers through the clients' optimal response function from Section~\ref{s:client}.

Our first result shows that \emph{if} a Nash equilibrium exists, then the corresponding marginal prices have to satisfy a nonlinear first-order ODE, derived from the dealers' first-order conditions for pointwise optimality. Note that unlike for the distributions with compact support considered in \cite{biais.al.00,biais.al.13}, there are no natural boundary conditions here.

\begin{lemma}
\label{lem:dealer:oli:unique}
Suppose Assumption \ref{ass:type} is satisfied. If ${\bf P_*} = (P_*, \ldots, P_*)$ is a Nash equilibrium for $K$ dealers, then the corresponding marginal prices satisfy the following ODE:
\begin{align}
P''(n) &=  \begin{cases}
\dfrac{(K-1) \gamma_c \left(\gamma_d n - P'(n) + g(P'(n) + \gamma_c K n)\right)}{ \frac{F}{f} \left(P'(n) + \gamma_c K n \right)- (\gamma_d n - P'(n) + g(P'(n) + \gamma_c K n))}, & n \in (-\infty, 0), \label{eq:lem:dealer:oli:unique:ODE}
 \\
\dfrac{(K-1) \gamma_c \left(\gamma_d n - P'(n) + g(P'(n) + \gamma_c K n)\right)}{ -\frac{\ol F}{f} \left(P'(n) + \gamma_c K n \right)- (\gamma_d n - P'(n) + g(P'(n) + \gamma_c K n))}, & n \in (0, \infty). 
\end{cases}
 \end{align}
 (In particular, it is part of the assertion that the denominators never vanish.)
\end{lemma}

We next establish existence and uniqueness of a solution to the ODE \eqref{eq:lem:dealer:oli:unique:ODE}. Compared to the monopolistic case from Section~\ref{ss:mon}, this requires an additional assumption. To wit, log-concavity of the client types as in Assumption \ref{ass:type} guarantees that $g'(y) < 1$ on $\RR$ by Proposition~\ref{prop:g:mon}. However, for well-posedness of the ODE, we need that this also remains true in the limit $y \to \pm \infty$.

\begin{assumption}
	\label{ass:g:oli}
The function $g$ from~\eqref{eq:condexp} satisfies $\limsup_{w \to \pm\infty} g'(w) < 1$.
\end{assumption}

 For normally distributed type variables s in Example~\ref{ex:normal} (or, more generally, if inventories and signals are generated from the same distributions as in Example~\ref{ex:same}),  $g(y)$ is just a linear function with slope equal to the projection coefficient. Therefore Assumption~\ref{ass:g:oli} is evidently satisfied in this case. 
 
\begin{theorem}
	\label{thm:ODE}
	Suppose Assumptions \ref{ass:type} and \ref{ass:g:oli} are satisfied. Then, there exists a unique function $P'_*: \RR \setminus \{0\} \to \RR$ that is increasing, continuously differentiable and satisfies the ODE~\eqref{eq:lem:dealer:oli:unique:ODE}.
\end{theorem}

\begin{remark}
Note that uniqueness for~\eqref{eq:lem:dealer:oli:unique:ODE} holds despite the absence of natural boundary conditions. For distributions with compact support, such boundary conditions are derived from a local analysis of the corresponding ODEs near the boundary points by~\cite{biais.al.00}. For distributions with support on the entire real line, monotonicity and finiteness of the solution suffice to guarantee uniqueness: there is only one value of the left and right derivatives at zero for which the solution remains increasing and  finite on the entire real line. 

\begin{figure}[htbp]
\begin{center}
\includegraphics[width=0.45\textwidth]{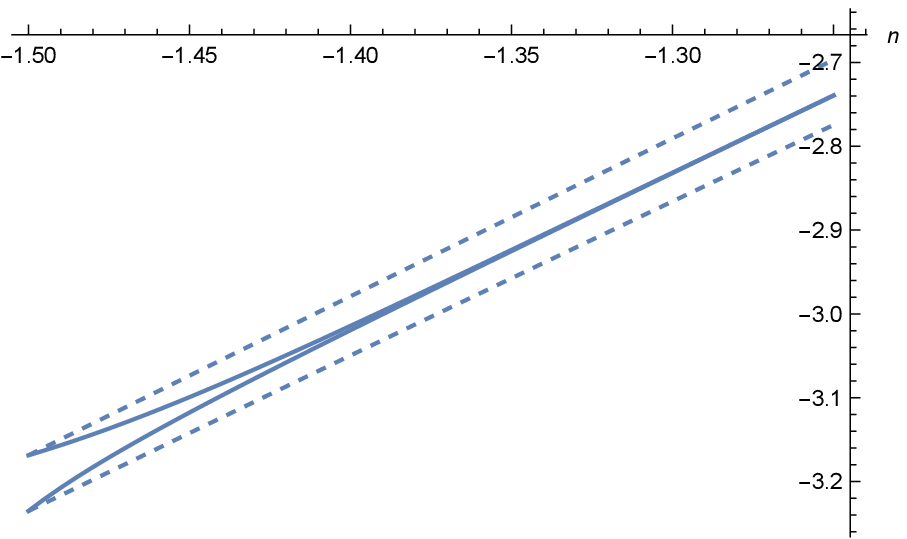}
\includegraphics[width=0.45\textwidth]{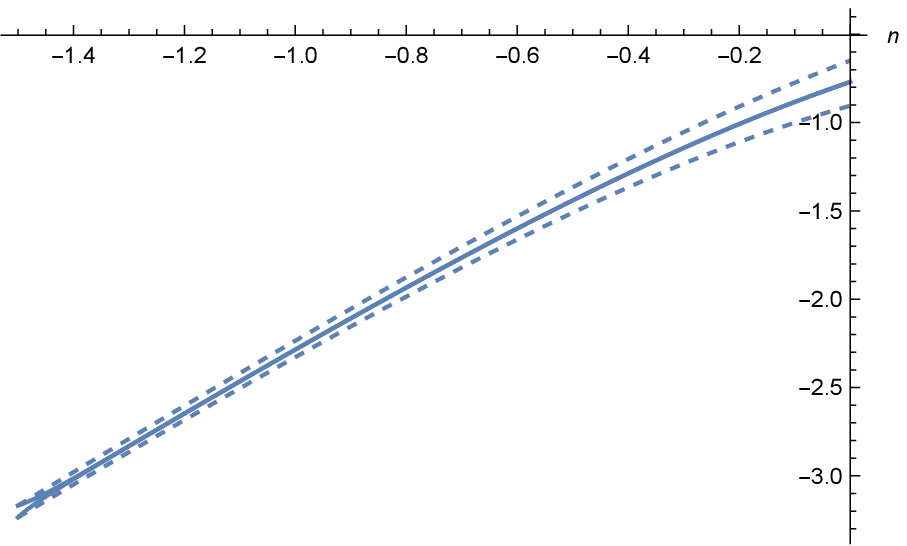}
\end{center}
\caption{Upper and lower solutions (dashed) and numerical solutions of the ODE~\eqref{eq:lem:dealer:oli:unique:ODE} starting from these boundary values for standard normal noise, client signals and inventories, $K=2$ dealers and inventory costs $\gamma_c=1$, $\gamma_d=0$.}
\label{fig:upperlower}
\end{figure}

Solving the equation numerically by a grid search for these derivatives is evidently extremely unstable. In contrast, the ODE can be solved in a stable manner (with concrete error bounds) by starting from the upper and lower solutions that we construct for the existence part of the proof of Theorem~\ref{thm:ODE}, and then solving the equation backwards, cf.~Remark~\ref{rem:recipe}. 
For standard normal primitives, this is illustrated in Figure~\ref{fig:upperlower}, which plots the numerical solutions of the ODE~\eqref{eq:lem:dealer:oli:unique:ODE} starting from the upper and lower solutions. Evidently, the convergence is very fast (as illustrated by the left panel that zooms in near the boundary points). Moreover, the right panel shows that the corresponding values of the solutions at zero (which are required for verifying Assumption~\ref{ass:f:oli}) are indistinguishable.
\end{remark}

Lemma~\ref{lem:dealer:oli:unique} and Theorem~\ref{thm:ODE} show that there is at most one symmetric Nash equilibrium. However, as pointed out by \cite{back.baruch.13}, existence of a well-behaved solution to the ODE~\eqref{eq:lem:dealer:oli:unique:ODE} is \emph{not} enough to identify a Nash equilibrium. The reason is that the ODE corresponds to the dealers' first-order conditions for pointwise optimality, which are not generally sufficient for global optimality here due to the absence of convexity in the dealers' optimization problems. 

As a way out, \cite{biais.al.13} impose additional restrictions on the primitives of the model that guarantee this convexity. Conditions of this type rule out standard examples like exponential or normal types, so we instead follow~\cite{back.baruch.13} in assuming that there is sufficient adverse selection in the market, in that the client's type has a sufficiently strong relation with the asset payoff. More specifically, our next result shows that if adverse selection \emph{and} the dealer's inventory costs (relative to the client's) are large enough, then the solution to the ODE~\eqref{eq:lem:dealer:oli:unique:ODE} indeed identifies the unique symmetric Nash equilibrium.

\begin{assumption}
	\label{ass:f:oli}
The distribution of the client type $f$ and the function $g$ from~\eqref{eq:condexp} satisfy
	\begin{align}
	& (F/f)' - g' \leq \frac{\gamma_d}{\gamma_c} \quad \text{on } (-\infty, P_*'(0-)] , \label{eq:ass:ex:f:f-}\\
	&(\bar F/f)' - g' \leq \frac{\gamma_d}{\gamma_c} \quad \text{on } [P_*'(0+), \infty), \label{eq:ass:ex:f:f+}
\end{align}
where $P'_*$ is the unique solution of the ODE~\eqref{eq:lem:dealer:oli:unique:ODE} from Theorem~\ref{thm:ODE}.
\end{assumption}

\begin{theorem}
	\label{thm:oli:ex}
Suppose Assumptions \ref{ass:type}, \ref{ass:g:oli} and \ref{ass:f:oli} are satisfied and let $P_*'$ be the solution of the ODE~\eqref{eq:lem:dealer:oli:unique:ODE} from Theorem~\ref{thm:ODE}. Then, ${\bf P_*} = (P_*, \ldots, P_*)$ with $P_*(n) = \int_0^n P'_*(x) \dd x$ is the unique symmetric Nash equilibrium for $K$ dealers. 
\end{theorem}

Assumption~\ref{ass:f:oli} is analogous to the conditions~\eqref{eq:rem:mon:est-}-\eqref{eq:rem:mon:est+} for convexity of the optimal price schedules quoted by a monopolistic dealer. The only difference is the range of marginal prices on which these constraints need to be imposed. In the monopolistic case, these are given by the range of (the inverse of) an explicit function; here they are instead determined by the solution of the ODE~\eqref{eq:lem:dealer:oli:unique:ODE}. Using the explicit upper and lower solutions from the existence proof in Theorem~\ref{thm:ODE} as boundary values, upper and lower bounds can readily be computed numerically using standard solvers for (uniformly Lipschitz) ODEs. 

Alternatively, sufficient conditions in terms of model primitives can be derived directly from upper and lower solutions of the ODE. For standard normal client inventories and signals and risk-neutral dealers, these conditions are satisfied if the projection coefficient $\beta$ from Example~\ref{ex:normal} is bigger than $0.55$. However, the numerical solution of the corresponding ODEs suggests that \eqref{eq:rem:mon:est-}-\eqref{eq:rem:mon:est+} are in fact already satisfied for $\beta >0.465$ for two dealers ($K=2$) and for $\beta \geq 0.5$ if the number of dealers is very large.

The bound $\beta \geq 0.5$ coincides with the sufficient condition of~\cite{back.baruch.13}, which can be derived by exploiting the specific properties of the normal distribution. To wit, in this case, the competitive price schedule of~\cite{glosten.89} yields a smaller upper solution on $(-\infty, 0)$ and a larger lower solution on $(0,\infty)$ for ODE~\eqref{eq:lem:dealer:oli:unique:ODE}, that is still known in closed form for normally-distributed types. This in turn provides an explicit sufficient condition for our general condition \eqref{eq:ass:ex:f:f-}--\eqref{eq:ass:ex:f:f+}. The same argument can also be applied to general non-centered Gaussian types as in Example~\eqref{ex:normal}: 

\begin{remark}\label{rem:normaoli}
For normally distributed types as in Example~\ref{ex:normal}, suppose that $\gamma_d = 0$ and
\begin{equation}
	\label{eq:rem:dealer:oli:norm:est}
	\tfrac{\gamma_c^2 \mu_M^2}{\sigma_S^2 + \gamma_c^2  \sigma_M^2} < \frac{\pi}{2} 
	\quad  \text{and} \quad \beta \geq \frac{1}{2} + \frac{1}{2}\frac{\Phi}{\phi} \big( z_{\mathrm{oli}}(\beta )\big)\tfrac{\gamma_c |\mu_M|}{\sqrt{\sigma_S^2 + \gamma_c^2}}.
\end{equation}
Here, $\Phi$ and $\phi$ denote the cumulative distribution function and the probability density function of the standard normal distribution, respectively, and 
\begin{equation*}
	z_{\mathrm{oli}}(\beta) := -\frac{(1- \beta)}{2 \beta -1}\tfrac{\gamma_c |\mu_M|}{\sqrt{\sigma_S^2 + \gamma_c^2  \sigma_M^2}}, \quad  \mbox{for } \beta > \frac{1}{2}.
\end{equation*}
Then Assumptions \eqref{eq:ass:ex:f:f-}--\eqref{eq:ass:ex:f:f+} are satisfied. In particular, if $\mu_M = 0$, then $z_{\mathrm{oli}}(\beta) = 0$ and \eqref{eq:rem:dealer:oli:norm:est} specialises to the sufficient condition $\beta \geq 1/2$ of~\cite{back.baruch.13}.\footnote{The boundary case $\beta = 1/2$ can be treated with a limiting argument.}
\end{remark}

For Gaussian primitives as in~\cite{glosten.89,back.baruch.13}, the left panel in Figure~\ref{fig:comp} displays the impact of the dealers' inventory costs on the optimal price schedules with competition. Unlike in the monopolistic case, bid-ask spreads no longer remain invariant but instead increase with the dealers' inventory costs. 

\begin{figure}
\begin{center}
\includegraphics[width=0.45\textwidth]{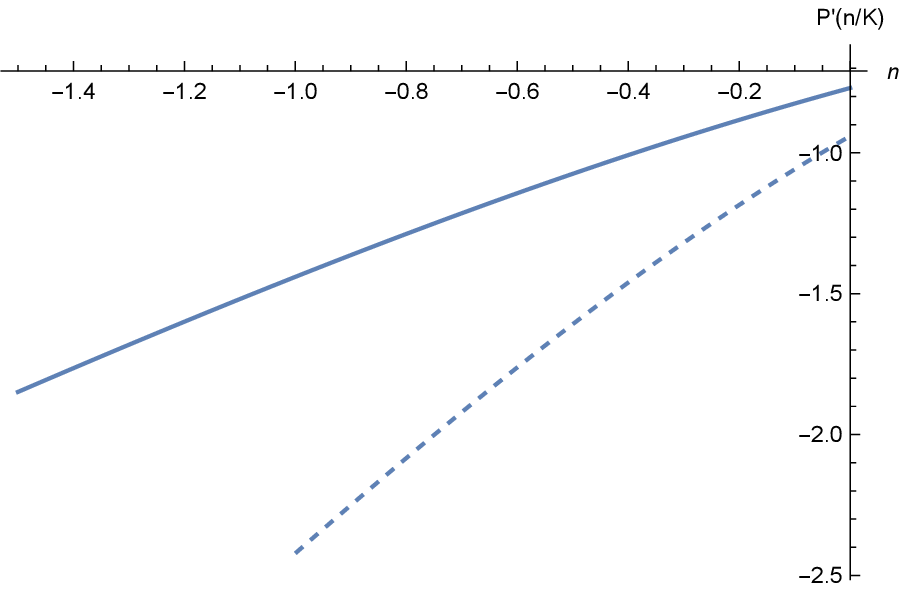}
\includegraphics[width=0.45\textwidth]{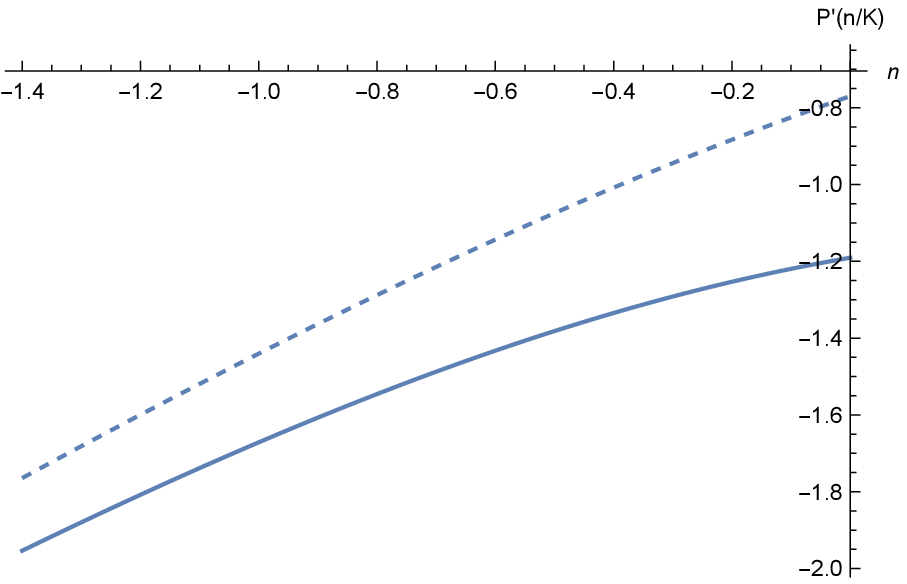}
\caption{Left panel: optimal marginal prices $P'(n/K)$ for $K=2$ dealers with inventory costs $\gamma_d=0$ (solid) and $\gamma_d=0.4$ (dashed). Right panel: marginal prices for $K=2$ competing dealers (solid) and a monopolist dealer (dashed)  with inventory costs $\gamma_d=0$. In each case, noise, client signals and inventories are standard normal ($\beta=0.5$) and the client's inventory cost is $\gamma_c=1$.}
\label{fig:comp}
\end{center}
\end{figure}

The right panel of Figure~\ref{fig:comp} compares the monopolistic marginal prices $P'_{\textrm{mon}}(n)$ from Lemma~\ref{lem:dealer:mon} to the marginal prices $P'_{\textrm{oli}}(n/K)$ of the aggregate price schedule $KP_{\textrm{oli}}(n/K)$ quoted by $K$ symmetric oligopolistic dealers, cf.~Theorem~\ref{thm:oli:ex}. We see that competition between the dealers leads to both tighter bid-ask spreads (as in \cite{ho.stoll.81}) and flatter price schedules (i.e., ``deeper'' markets). Both properties follow directly from the fact that, by the proof of Theorem \ref{thm:ODE}, the marginal prices $P'_{\textrm{mon}}(n)$ of the monopolist on $(-\infty, 0)$ are a lower solution to the ODE for the marginal prices $P'_{\textrm{oli}}(n/K)$ quoted in the oligopolistic case.
		
\appendix

\section{Proofs for Section~\ref{s:client}}

\begin{proof}[Proof of Lemma~\ref{lem:adm schedule}]
Since all dealers quote the same price schedule in the context of this lemma, the client's goal functional~\eqref{eq:client} simplifies. Indeed, the client's optimal trade $n^P(y)$, $y \in \mathbb{R}$ then is the unique maximizer of the scalar function 
\begin{equation}\label{eq:HP}
n \mapsto H^P(y; n) :=\frac{1}{K} J^P_c(y; n, \ldots, n) =  y n -P(n)  - \frac{K \gamma_c}{2} n^2.
\end{equation}

(a) Fix $y_{(0)} \in \RR$ and let $(y_{(i)})_{i \in \NN}$ be a sequence in $\RR$ with $\lim_{i \to \infty} y_{(i)} = y_{(0)}$. For $i \in \NN_0$, set $n_{(i)} := n^P(y_{(i)})$. We have to show that $\lim_{i \to \infty} n_{(i)} = n_{(0)}$. Denote by $\cA$ the set of all accumulation points in $[-\infty, \infty]$ of the sequence $(n_{(i)})_{i \in \NN}$. As $\cA \neq \emptyset$, it suffices to show that $\cA \cap \left([-\infty, \infty] \setminus \{n_{(0)}\}\right) = \emptyset$.

First, we show that  $-\infty, \infty \notin \cA$. We only establish the claim for $\infty$; the corresponding assertion for $-\infty$ follows from a similar argument. Seeking a contradiction, suppose that $\infty \in \cA$. Then there exists a subsequence, again denoted by $(n_{(i)})_{i \in \NN}$ for convenience, such that $\lim_{i \to \infty} n_{(i)} = \infty$. By maximality of each $n_{(i)}$ for $H^P(y_{(i)}; \cdot)$ and continuity of $H^P(\cdot,n_{(0)})$, 
\begin{equation*}
\liminf_{i \to \infty} H^P\left(y_{(i)}; n_{(i)}\right) \geq \liminf_{i \to \infty} H^P\left(y_{(i)}; n_{(0)}\right) = H^P\left(y_{(0)}; n_{(0)}\right).
\end{equation*}
Together with maximality of $n^P(y_{(0)}+ 1)$ for $H^P(y_{(0)} + 1; \cdot)$ and the definition of $H^P$ in~\eqref{eq:HP}, this leads to the desired contradiction:
\begin{align*}
H^P\left(y_{(0)}+1; n^P(y_{(0)} + 1)\right) &\geq \liminf_{i \to \infty} H^P\left(y_{(0)}+1; n_{(i)}\right) \\
&= \liminf_{i \to \infty} \left(n_{(i)} \left(1 + y_{(0)} - y_{(i)}\right)   + H^P\left(y_{(i)}; n_{(i)}\right) \right) \\
&\geq \liminf_{i \to \infty} \left(n_{(i)} \left(1 + y_{(0)} - y_{(i)}\right)\right) + H^P\left(y_{(0)}; n_{(0)}\right) \\
& =\infty.
\end{align*}

It remains to show that $ \cA \cap (\RR \setminus \{n_{(0)}\}) = \emptyset$. Seeking a contradiction, suppose there is $\bar n \in \cA \cap  (\RR \setminus \{n_{(0)}\}) $.  Then there exists a subsequence, again denoted by $(n_{(i)})_{i \in \NN}$ for convenience, such that $\lim_{i \to \infty} n_{(i)} = \bar{n}$. Recall that $n_{(0)}$ is the \emph{unique} maximizer of $H^P(y_{(0)}; \cdot)$, $H^P$ is continuous in both variables, and each $n_{(i)}$ is maximal for $H^P(y_{(i)};\cdot)$. Whence, we again arrive at a contradiction:
\begin{equation*}
H^P\left(y_{(0)}; n_{(0)}\right) > H^P\left(y_{(0)}; \ol n\right) = \lim_{i \to \infty} H^P\left(y_{(i)}; n_{(i)}\right) \geq \liminf_{i \to \infty} H^P\left(y_{(i)}; n_{(0)}\right) = H^P\left(y_{(0)}, n_{(0)}\right).
\end{equation*}
We conclude that the function $n^P(y)$ is continuous as asserted.

(b) To prove this part of the lemma, we only show the following formally weaker claim:

\medskip

\noindent (b') There are $a^P \in (-\infty, \infty]$ and $b^P \in [-\infty,\infty)$ such that $I^P = (-\infty, a^P) \cup (b^P, +\infty)$.

\medskip

We will then use only (b') to establish (c) and argue in (c) that  $a^P \leq b^P$. Since $H^P(y, 0) = 0$, we have $I^P = I^P_+ \cup I^P_-$, where
\begin{align*}
I^P_+ &:= \left\{y \in \RR: \text{there is } n > 0 \text{ with } H^P(y; n) > 0\right\}, \\
I^P_- &:= \left\{y \in \RR: \text{there is } n < 0 \text{ with } H^P(y; n) > 0\right\}.
 \end{align*}
These sets are both nonempty, because
\begin{equation}
\label{eq:lim m infty}
\lim_{y \to \infty} H^P(y; 1) = \infty \quad \text{and}  \quad \lim_{y \to -\infty} H^P(y, -1) = \infty.
\end{equation}
Observe that if $y \in I^P_-$, then $y' \in I^P_-$ for all $y' \leq y$ and also for all $y' > y$ in a (sufficiently small) neighbourhood of $y$. Set $a^P := \sup\{y \in \RR: y \in I_-^P\} \in (-\infty, \infty]$. The previous argument implies that $I_-^P = (-\infty, a^P)$. Similarly, set $b^P := \inf\{y \in \RR: y \in I_+^P\} \in [-\infty, \infty)$. Then $I^P_+ = (b_P, \infty)$.

(c) For $y \in I^P$, we have $n^P(y) \neq 0$ so that $H^P(y,\cdot)$ is differentiable in $n$ at $n^P(y)$. Hence, maximality of $n^P(y)$ for $H^P(y, \cdot)$ implies the following first-order condition (FOC): 
\begin{equation}
\label{eq:lem schedule 1d:FOC}
\frac{\diff H^P}{\diff n}\left(y; n^P(y)\right) = 0, \quad y \in I^P.
\end{equation}
We now use this to show that for all $\ol y \in I^P$, there is an open neighbourhood $U_{\ol y} \subset I^P$ of $\ol y$ such that for all $y \in U_{\ol y}$,
\begin{equation}
n^P(y)
\begin{cases}
< n^P(\ol y) & \text{if } y < \ol y, \\
> n^P(\ol y)  & \text{if } y > \ol y.
\end{cases}
\label{eq:loc inc}
\end{equation}
Fix $\ol y \in I^P$ and set $\ol n := n^P(\ol y)$. As $\ol n$ is the unique maximum of $H^P(\ol y; \cdot)$, there exists an open neighbourhood
$U_{\ol n}$ of $\ol n$ such that, for all $n \in U_{\ol n}$, we have
\begin{equation}
\label{eq:FOC}
\frac{\diff H^P}{\diff n}(\ol y; n) 
\begin{cases}
> 0 & \text{if } n < \ol n, \\
< 0 & \text{if } n > \ol n.
\end{cases}
\end{equation}
As $n^P$ is continuous, there exists an open neighbourhood $U_{\ol y} \subset I^P$ of $\ol y$ such that $n^P(U_{\ol y}) \subset U_{\ol n}$. Now, if $y \in U_{\ol y}$ with $y < \ol y$, then \eqref{eq:lem schedule 1d:FOC} and the definition of $H^P$ give
\begin{equation*}
\frac{\diff H^P}{\diff n}\left(\ol y; n^P(y)\right)= \frac{\diff H^P}{\diff n}\left(\ol y; n^P(y)\right) - \frac{\diff H^P}{\diff n}\left(y; n^P(y)\right) = \ol y - y > 0.
\end{equation*}
Together with \eqref{eq:FOC}, this implies $n^P(y) < \ol n$. Similarly, if $y \in U_{\ol y}$ with $y > \ol y$, then $n^P(y) > \ol n$. Therefore, \eqref{eq:loc inc} indeed holds on a suitable neighbourhood of $\bar{y}$.

We proceed to show that \eqref{eq:loc inc} together with (b') and Lemma~\ref{lem:increasing} implies that $n^P$ is increasing on $I^P$. First, if $a^P > b^{P}$, then $I^P = \RR$ and the claim follows directly from Lemma~\ref{lem:increasing}. Otherwise, if $a^P \leq b^P$, Lemma~\ref{lem:increasing} implies that $n^P$ is increasing on $I^P_+ = (-\infty, a^P)$ and on $I^P_- = (b^P, \infty)$. Since $n^P$ is zero on $\RR \setminus (I^P_+ \cup I^P_-) = [a^P, b^P]$ and continuous on $\RR$, it follows that  $n^P$ is negative on $(-\infty, a^P)$ and positive on $(b^P, \infty)$. Hence, $n^P$ is also increasing on $I^P$.

Finally, we show that $\lim_{y \to \infty} n^P(y) = \infty$ and $\lim_{y \to -\infty} n^P(y) = -\infty$. Recall that $n^P$ is continuous on $\RR$ (cf.~(a)), increasing and nonzero on $I^P$ (as shown above) and zero on $\RR \setminus I^P$ by definition. Therefore, the limits at $\pm\infty$ in turn yield that $a^P \leq b^P$,  $n^P(I^P) = \RR \setminus \{0\}$ and $\lim_{y \uparrow a^P} n^P(y) = 0 =  \lim_{y \downarrow b^P} n^P(y)$.

We only spell out the argument for $\lim_{y \to -\infty} n^P(y) = -\infty$; the corresponding argument for $\lim_{y \to \infty} n^P(y) = \infty$ is similar. Set $n(-\infty) := \lim_{y \to -\infty} n^P(y)$. As $n^P$ is increasing on $I^P$ and $y \in I^P$ for $y < a^P$, it follows that $n(-\infty)$ is well defined and valued in $[-\infty,\infty)$. Seeking a contradiction, suppose that $n(-\infty) > -\infty$. We distinguish two cases. First, assume that $n(-\infty) \geq 0$. Then $n(y) > 0$ for all $y  \in I^P$. In particular for all $y < \min(a^P, 0)$, by maximality of $n^P(y)$ for $H^P(y;\cdot)$ and the definition of $H^P$, we obtain
\begin{equation*}
H^P(y; -1) \leq H^P\left(y; n^P(y)\right) \leq H^P\left(0; n^P(y)\right) \leq H^P\left(0, n^P(0)\right).
\end{equation*}
Hence, it follows that
\begin{equation*}
\lim_{y \to -\infty}H^P(y; -1) \leq H^P\left(0, n^P(0)\right),
\end{equation*}
which contradicts~\eqref{eq:lim m infty}. Next, assume that $n(-\infty) \in (-\infty, 0)$. Then there exists $\ol y < a^P$ such that $n^P(y) < 0$ for all $y \leq \ol y$. By the FOC~\eqref{eq:lem schedule 1d:FOC}, we have
\begin{equation*}
K \gamma_c n^P(y) + P'\left(n^P(y)\right) = y, \quad  \mbox{for $y \leq \ol y$}.
\end{equation*}
As $P'$ is continuous on $(-\infty,0)$, taking limits as $y \to -\infty$ yields
\begin{equation*}
K\gamma_c n(-\infty) + P'(n(-\infty))=-\infty.
\end{equation*}
This contradicts $n(\infty) \in (-\infty, 0)$. In summary, we conclude that $n(\infty)=\infty$ as claimed.
\end{proof}

\begin{proof}[Proof of Lemma~\ref{lem:schedules:hom}] 
(a) ``$\Rightarrow$'': Suppose $P$ is admissible for $K=1$ dealer. To establish strict convexity of $n \mapsto \frac{\gamma_c}{2} n^2 + P(n)$, it suffices to show that $n \mapsto \gamma_c n + P'(n)$ is increasing on $\RR \setminus \{0\}$. So fix $n_1, n_2 \in \RR \setminus \{0\}$ with $n_1 < n_2$. Recall that by the first-order condition~\eqref{eq:lem schedule 1d:FOC},
$$0=\frac{\diff }{\diff n}H^P\left((n_{P})^{-1}(n_i); n_i\right) = (n_{P})^{-1}(n_i)-P'(n_i)-\gamma_c n_i, \quad \mbox{for $i \in \{1, 2\}$,}$$
and $(n_{P})^{-1}$ is increasing on $\RR \setminus \{0\}$  by Lemma~\ref{lem:adm schedule}(c). As a consequence,
\begin{equation*}
	\gamma_c n_1 + P'(n_1) - (\gamma n_2 + P'(n_2)) = (n_{P})^{-1}(n_1) - (n_{P})^{-1}(n_2) < 0.
\end{equation*}
Hence, $n \mapsto \gamma_c n + P'(n)$ is increasing on  $\RR \setminus \{0\}$, and  $n \mapsto \frac{\gamma_c}{2} n^2 + P(n)$ is in turn strictly convex. Finally, as $\gamma_c n +P'(n) = (n_{P})^{-1}(n)$ by \eqref{eq:lem schedule 1d:FOC} for all $n \in \NN$, Lemma~\ref{lem:adm schedule}(c) yields $\lim_{n \to \pm\infty} \left(\gamma_c n +P'(n)\right) = \pm\infty$.

Conversely, suppose $n \mapsto \frac{\gamma}{2} n^2 + P(n)$ is strictly convex on $\RR$ with  $\lim_{n \to \pm\infty} \left(\gamma_c n +P'(n)\right) = \pm\infty$. By \cite[Proposition B.22(b)]{bertsekas.99}, the function $P$ has left and right derivatives for all $n \in \RR$, denoted by $P'(n-)$ and  $P'(n+)$ (which coincide with $P'(n)$ for $n \neq 0$). Moreover, for each fixed $y \in \RR$, the function $H^{P}(y;\cdot)$ is strictly concave on $\RR$. Hence it has at most one maximum, and by \cite[Proposition B.24(f)]{bertsekas.99},
$n$ is a maximum of $H^{P}(y;\cdot)$ if and only if
\begin{equation*}
	0 \in \partial_n H^{P}(y; n) = \left[y- \gamma_c n - P'(n+), y - \gamma_c n - P'(n-)\right].
\end{equation*}
Here, $\partial_n H^{P}(y; n)$ denotes the subdifferential of  $H^{P}(y; \cdot)$ at $n$.
Concavity of the functions $H^{P}(y;\cdot)$, continuous differentiability of $P$ on $\mathbb{R} \setminus \{0\}$, and $\lim_{n \to \pm\infty} \left(\gamma_c n +P'(n)\right) = \pm\infty$ imply that $\bigcup_{n \in \RR} \partial_n H^{P}(y; n) = \RR$ for each $y \in \RR$. Hence, for each $y \in \RR$, there exists $n \in \RR$ such that $0 \in \partial_n H_{P}(y;n)$.

(b)	``$\Rightarrow$'': Let $P$ be an admissible price schedule for $K \geq 2$ dealers. We first show that $P$ is convex. To this end, it suffices to check that $P''(n) \geq 0$ for all $n \in \RR \setminus \{0\}$. So fix $n \in \RR \setminus \{0\}$ and set $y := (n^P)^{-1}(n) \in I^P$, so that ${\bf n} := (n, \ldots, n) \in \mathbb{R}^K$ is the unique maximum of the function $J^P_c(y; \cdot)$. As $n \neq 0$, $J^P_c(y; \cdot)$ is twice differentiable in $\mathbf{n}$ and the second-order necessary optimality condition in turn implies that
	\begin{equation*}
	\nabla^2_{\bf n} J^P_c\left(y; {\bf n}\right) = -\gamma_c \ones \ones^\top- \diag (P''(n), \ldots P''(n)) \;\mbox{is negative semidefinite}.
	\end{equation*}
	(Here, $\ones =(1, \ldots, 1)^\top \in \mathbb{R}^K$.) Hence, all eigenvalues of $\gamma_c \ones\ones^\top + \diag (P''(n), \ldots P''(n))$ are nonnegative. Using the matrix determinant lemma, it is not difficult to verify that the matrix $\gamma_c \ones \ones^\top + \diag (P''(n), \ldots P''(n))$ has the eigenvalue $K \gamma_c + P'(n)$ with algebraic multiplicity $1$ and the eigenvalue $P''(n)$ with algebraic multiplicity $K-1$. Whence, for $K \geq 2$, we have $P''(n) \geq 0$ so that $P$ is indeed convex.
	
	We proceed to show that the price schedule $P$ is even strictly convex. To this end, it suffices to show that $P'$ is increasing on $\RR \setminus \{0\}$. Seeking a contradiction, suppose that there are $n_1, n_2 \in \RR \setminus \{0\}$ with $n_1 < n_2$ such that $P'(n_1) \geq P'(n_2)$. Since $P'$ is nondecreasing on $\RR \setminus \{0\}$ by convexity of $P$ on $\RR$, it follows that $P'(n) = P'(n_1) = P'(n_2)$ for all $n \in [n_1, n_2] \setminus \{0\}$. Set
	\begin{equation*}
	\tilde n_2 :=
	\begin{cases}
	n_2 & \text{if } n_1 > 0, \\
	\frac{n_1}{2} &\text{if } n_1 < 0.
	\end{cases}
	\end{equation*}
	Then $\tilde n_2 > n_1$, we have $[n_1, \tilde n_2] \subset \RR \setminus \{0\}$, and $P'$ is constant on $[n_1, \tilde n_2]$. As a consequence, $P(n) = P(n_1) + P'(n_1) (n-n_1)$ for all $n \in [n_1, \tilde n_2]$. Let $\ol n \in (n_1, \tilde n_2)$, choose $\epsilon > 0$ such that $\ol n -\epsilon, \ol n +\epsilon \in [n_1, \tilde n_2]$, and set $\ol y := (n_P)^{-1}(\ol n)$. Then
	\begin{align*}
	J^P_c(\ol y; \ol n, \ldots, \ol n) &= J^P_c(\ol y; \ol n -\epsilon, \ol n+ \epsilon, \ol n, \ldots, \ol n).
	\end{align*}
	Hence, the function $J^P_c(\ol y; \cdot)$ has at least two maximizers, contradicting the admissibility of $P$.
	
	``$\Leftarrow$'':  this part of the assertion follows from Theorem~\ref{thm:schedules het}, which treats the more general case of not necessarily symmetric price schedules. (Note that the proof of this result only uses the ``$\Rightarrow$'' direction of Lemma~\ref{lem:schedules:hom}, which has already been established.)
\end{proof}

\begin{proof}[Proof of Theorem~\ref{thm:schedules het}]
First, assume that $P_1, \ldots, P_K$ are compatible. Then there is $\ol y \in \mathbb{R}$ such that the client's problem $\max_{\mathbf{n}} J^{\mathbf{P}}(\ol y; \mathbf{n})$ has a solution ${\bf \ol n} = (\ol n_1, \ldots, \ol n_K)$. Since the price schedules $P_1, \ldots, P_K$ are admissible, they are strictly convex by Theorem \ref{lem:schedules:hom}(b). As a consequence, the client's (normalized) goal function $J^{\bf P}(\ol y, \cdot)$ is strictly concave and therefore has only one maximizer. Moreover, ${\bf \ol n} \in \RR^K$ is a maximizer of $J^{\bf P}(\ol y, \cdot)$ if and only if ${\bf 0} \in \partial_{\bf n} J^{\bf P}(\ol y, {\bf \ol n})$. Hence, for each $k \in \{1, \ldots, K\}$, we have
\begin{equation*}
	{\bf 0} \in \left(\partial_{\bf n} J^{\bf P}(\ol y, {\bf \ol n})\right)_k = \left[ \ol y - \gamma_c \sum_{i=1}^K \ol n_i - P'_k(\ol n_k+),\ol y - \gamma_c \sum_{i=1}^K \ol n_i - P'_k(\ol n_k-)\right].
\end{equation*}
Since $P_k'$ is increasing on $\mathbb{R} \setminus \{0\}$ by strict convexity of $P_k$, this is equivalent to
\begin{equation*}
	\ol y - \gamma_c \sum_{i=1}^K \ol n_i \in \left[ P'_k(\ol n_k-), P'_k(\ol n_k+)\right] \subset \left(\ell_k, r_k\right).
\end{equation*}
As this holds for any $k \in \{1,\ldots,K\}$, it follows that $\bigcap_{k=1}^K (\ell_k, r_k) \neq \emptyset$. This in turn yields \eqref{eq:ol:lr}.

\medskip{}
Conversely, assume that \eqref{eq:ol:lr} is satisfied. We proceed to show that then (b) and (c) and in turn (a) are satisfied. The latter also implies a fortiori that $P_1, \ldots, P_K$ are compatible.

For $k \in \{1, \ldots, K\}$, set $p_k := P_k'$ and denote its inverse function by $p_k^{-1}$, with the convention that $p_k^{-1}(m) = 0$ if $m \in [p_k(0-), p_k(0+)].$ Then each $p_k^{-1}$ is nondecreasing on  $(\ol\ell, \ol r)$ and increasing on  $(\ol\ell, \ol r) \setminus [p_k(0-), p_k(0+)]$. Hence, the function $y^{\bf P}: (\ol\ell, \ol r)\to \RR$ defined by
\begin{equation*}
	y^{\bf P}(m) := m + \gamma_c \sum_{k =1}^K p^{-1}_k(m),
\end{equation*}
is continuous and increasing. Moreover, it satisfies $\lim_{m \to \ol\ell} y^{\bf P}(m) = -\infty$ and $\lim_{m \to \ol r} y^{\bf P}(m) = \infty$. To wit, there exists at least one $k_1 \in \{1, \ldots, K\}$ such that $\ell_{k_1} = \ol\ell$ and at least one $k_2 \in \{1, \ldots, K\}$ such that $r_{k_2} = \ol r$. For these indices, we then have $\lim_{\ell \to  \ell_{k_1}} p^{-1}_{k_1}(\ell) = -\infty$ and $\lim_{\ell \to  r_{k_2}} p^{-1}_{k_2}(\ell) = \infty$ by the definition of $\ell_{k_1}$ and $\ell_{k_2}$ in~\eqref{eq:lr}. Now, for $k =1,\ldots,K$, define the functions $n_k^{\bf P}: \RR \to \RR$ by
\begin{equation*}
	n_k^{\bf P}(y) := p^{-1}_k \left((y^{\bf P})^{-1}(y)\right).
\end{equation*}
Then, in view of the properties of the functions $p^{-1}_k$ and $y^{\bf P}$ established above, $n_k^{\bf P}$ is continuous, nondecreasing on $\RR$ and increasing on $I^{\bf P}_k$. So we have (c). This in turn implies that 
$$n_k^{\bf P}(I^{\bf P}_k) = \left(\lim_{y \to -\infty} p^{-1}_k \left((y^{\bf P})^{-1}(y)\right), \lim_{y \to \infty} \left((y^{\bf P})^{-1}(y)\right) \right) = (p^{-1}_k(\ol \ell), p^{-1}_k(\ol r))$$
and we have (b). To complete the proof, it now remains to establish (a). Let $\ol y \in \mathbb{R}$. We need to establish that $\ol{\bf n} = (\bar{n}_1,\ldots,\bar{n}_K):=  (n_1^{\bf P}(\ol y), \ldots, n_K^{\bf P}(\ol y))$ is the maximizer of $J^{\bf P}(\ol y, \cdot)$. Uniqueness follows as in the first part of the proof by strict convexity of $P_1, \ldots, P_K$. By \cite[Proposition B.24(f)]{bertsekas.99}, it suffices show that $ {\bf 0} \in \partial_{\bf n} J^{\bf P}(\ol y; {\bf \ol n})$. So fix $k \in \{1, \ldots, K\}$. We first consider the case $\ol n_k \neq 0$. Then, the definition of $n_k^{\bf P}$ and $P'_k(\ol n_k) = (y^{\bf P})^{-1}(\ol y)$ give
\begin{align*}
	\left(\partial_{\bf n} J^{\bf P}(\ol y; {\bf \ol n})\right)_k &= \left\{\ol y - \gamma_c \sum_{i=1}^K p^{-1}_k((y^{\bf P})^{-1}(\ol y)) - (y^{\bf P})^{-1}(\ol y)\right\} = \left\{\ol y -  y^{\bf P}((y^{\bf P})^{-1}(\ol y)) \right\} = \{0\},
\end{align*}
so that $0 \in \left(\partial_{\bf n} J^{\bf P}(\ol m, {\bf \ol n})\right)_k $. Next, we turn to the case $\ol n^k = 0$. Then,
\begin{equation*}
	\left(\partial_{\bf n} J^{\bf P}(\ol y; {\bf \ol n})\right)_k = \left[\ol y - \gamma_c \sum_{i=1}^K \ol n_i - p_k(0+),\ol y - \gamma_c \sum_{i=1}^K \ol n_i - p_k(0-)\right].
\end{equation*}
The definition of $y^{\bf P}$ in turn yields that zero is an element of this subdifferential also in the second case,
\begin{equation*}
	0 = \ol y - \gamma_c \sum_{i=1}^K p^{-1}_k((y^{\bf P})^{-1}(\ol y)) - (y^{\bf P})^{-1}(\ol y) \in \left(\partial_{\bf n} A_{\bf P}(\ol m, {\bf \ol n})\right)_k.
\end{equation*}
Here, the set membership follows from the definition of $n^{\bf P}_k$ and from $(y^{\bf P})^{-1}(\ol y) \in [p_k(0-), p_k(0+)]$ (which holds by definition of $p^{-1}_k$ because $\ol n_k=0$). 
\end{proof}

\begin{example}\label{ex:counter}
	The following price schedules are admissible for $K=2$ dealers, but not compatible:
	\begin{equation*}
		P_1(x) := \int_0^x \arctan(y) \dd y, \quad P_2(x) := \int_0^x (\arctan(y) + \pi) \dd x.
	\end{equation*}
	Indeed, $(r_1, \ell_1) = (-\frac{\pi}{2}, \frac{\pi}{2})$ and $(r_2, \ell_2) = (\frac{\pi}{2}, \frac{3}{2} \pi)$, so that $P_1, P_2$ are not compatible by Theorem~\ref{thm:schedules het}. Analogous counterexamples can be constructed for distributions with compact support, which shows that a compatibility condition for unilateral deviations also needs to be imposed in the setting of~\cite{biais.al.00}.
\end{example}


For the proofs of Lemmas~\ref{lem:dealer:mon} and \ref{lem:dealer:oli:unique}, we need the following estimates:

	\begin{lemma}
		\begin{enumerate}
			\item Let $P$ be an admissible price schedule for $K=1$ dealer. Then:
			\begin{equation}
				\label{eq:lem:schedule:1:estimate}
				|n^P(y)| \frac{P'(0-) + P'(0+)}{2} -\frac{\gamma_c}{2} n^P(y)^2 \leq P(n^P(y)) \leq \frac{1}{2\gamma_c} y^2.
			\end{equation}
\item Let $P_1, \ldots, P_K$ be admissible price schedules for $K \geq 2$ dealers that are compatible and set ${\bf P} := (P_1, \ldots, P_K)$. Then for all $k \in \{1, \ldots, K\}$ and $y \in \RR$,\footnote{Here, we use (as always) the convention that $(P'_k)^{-1}(0) = 0$ if $0 \in [P'_k(0-), P'_k(0+)]$.}
			\begin{equation}
				|n^{\bf P}_k(y)| \frac{P'_k(0-) + P'_k(0+)}{2}  \leq P_k(n^{\bf P}_k(y)) \leq \frac{1}{2\gamma} y^2 + \frac{\gamma}{2} \bigg(\sum_{j \neq K} (P'_j)^{-1}(0) \bigg)^2.
				\label{eq:lem:oli:estimate}
			\end{equation}
		\end{enumerate}
	\end{lemma}
	
	\begin{proof}
	(a)	Fix $y \in \RR$. It suffices to consider the case $n^P(y) \neq 0$. We only consider the case $n^P(y) > 0$, the case $n^P(y) < 0$ can be argued similarly. 
		Using that the function $n \mapsto \gamma_c n + P'(n)$ is increasing on $\RR \setminus \{0\}$ by Lemma \ref{lem:adm schedule}(a), it follows from the mean value theorem that
		\begin{equation}
			P'(0+) n^P(y) \leq \frac{\gamma_c}{2} n^P(y)^2 + P(n^P(y))\leq \Big(\gamma_c n^P(y) + P'(n^P(y))\Big) n^P(y). \label{eqn:estimate:lem:schedule hom 1:estimate}
		\end{equation}
		Now, the lower bound in \eqref{eq:lem:schedule:1:estimate} follows from the first inequality in \eqref{eqn:estimate:lem:schedule hom 1:estimate} together with $\frac{P'(0-) + P'(0+)}{2} \leq P'(0+)$.
		Finally, using that $ \gamma_c n + P'(n) = (n^P)^{-1}(n)$ for $n \neq 0$ and then rearranging the second inequality in  \eqref{eqn:estimate:lem:schedule hom 1:estimate} gives
		\begin{equation*}
			P(n^P(y))\leq y n^P(y) - \frac{\gamma_c}{2} n^P(y)^2.
		\end{equation*}
		Now the upper bound in~\eqref{eq:lem:schedule:1:estimate} follows from the elementary inequality $a b \leq \frac{1}{2 \gamma_c} a^2 + \frac{\gamma_c}{2} b^2$ for $a, b \in \RR$.
		
	(b) The argument is very similar to the proof of part (a). The only difference is that we now use that 
	\begin{equation*}
		\gamma n^{\bf P}_k(y)  + P_k'(n^{\bf P}_k(y) ) +\gamma \bigg(\sum_{j \neq k} (P'_j)^{-1}(P_k'(n^{\bf P}_k(y) )) \bigg) = y. \qedhere
	\end{equation*}
	\end{proof}

\section{Proofs for Section~\ref{s:nash}}

	\begin{proof}[Proof of Proposition~\ref{prop:g:mon}]
	Recall that $V = S + \epsilon$ and $Y = S - \gamma M$, where $\epsilon, S, M$ are independent. Because the error term $\epsilon$ has mean zero, it follows that $g(y) = E[V\,|\,Y = y] = E[S\,|\,Y = y]$. Since $\id(y) - g(y) = y -  E[S\,|\,Y = y] = E[-\gamma_c M \,|\,Y = y]$, it therefore suffices to show that $y \mapsto E[S\,|\,Y = y]$ as well as $y \mapsto E[-\gamma_c M \,|\,Y = y]$ are continuously differentiable with positive derivatives. This follows from  Proposition \ref{prop:cond:mean}.
\end{proof}

\subsection{Proof of Lemma ~\ref{lem:dealer:mon}}

In this section, we prove Lemma~\ref{lem:dealer:mon} about the monopolistic dealer's optimal price schedule. We do this under weaker (but substantially less intuitive) assumptions than the convenient sufficient condition imposed in Assumption~\ref{ass:type}.

\begin{assumption}\label{ass:g}
	The expected payoff of the risky asset conditional on the clients' type,
	$y \mapsto g(y)=E[V\,|\,Y = y]$
	is continuously differentiable and of linear growth.
\end{assumption}

\begin{assumption}
	\label{ass:mon}
The probability density function $f$ of $Y$ is positive on $\RR$ and satisfies $\int_{-\infty}^\infty y^2 f(y) \dd y < \infty$. Moreover, setting
	\begin{align*}
		y_- &:= \inf\{y \in \RR: F(y)/f(y)+ y - g(y) = 0\},\\
		y_+ &:= \sup\{y \in \RR: -\ol F(y)/f(y)+ y - g(y) = 0\},
	\end{align*}
	we have:
	\begin{enumerate}
		\item[(i)] $(F/f) +\id - g$ is continuously differentiable on $(-\infty,y_-)$ with positive derivative and nonnegative on $[y_-, \infty)$;
		\item[(ii)] $-(\ol F/f) +\id-g$ is continuously differentiable on  $(y_+, \infty)$ with positive derivative and nonpositive on $(-\infty, y_+]$.
	\end{enumerate}
\end{assumption}

\begin{remark}
Note that since $(F/f) +\id - g > -(\ol F/f) +\id-g$, Assumption \ref{ass:mon} implies in particular that $y_- < y_+$. It is straightforward to check using Proposition~\ref{prop:g:mon} that log-concavity as in Assumption~\ref{ass:type} indeed implies Assumptions~\ref{ass:mon} and \ref{ass:g}. When the conditional mean function $g$ as well as the probability density and cumulative distribution functions $f, F$ of the client's type are available, e.g., for Gaussian types, then $y_\mp$ can be computed in a straightforward manner as the roots of explicit scalar functions.
\end{remark}

\begin{proof}[Proof of Lemma~\ref{lem:dealer:mon}]
We prove the result under the weaker Assumptions  \ref{ass:g} and \ref{ass:mon}. In light of the square-integrability of $f$ and the estimate \eqref{eq:lem:schedule:1:estimate}, for all admissible schedules we have $J_d(P) < \infty$ and $J_d(P) > -\infty$ if and only if $\int_{-\infty}^\infty n^P(y)^2f(y) \dd y < \infty$. So fix an admissible schedule $P$ such that $\int_{-\infty}^\infty n^P(y)^2f(y) \dd y < \infty$ and let $a^P$ and $b^P$ be as in Lemma \ref{lem:adm schedule}(b). Define the function $h: \RR \to \RR$ by $h:= f g$ and set $H(x) := \int_{-\infty}^x h(y) \dd y$ and $\ol H(x) := \int_{x}^\infty h(y) \dd y$. Moreover, denote the inverse function of $n^P$ on $(-\infty, a^P) \cup (b^P, \infty)$ by $y^P$, set $p := P'$ and note that 
\begin{equation}
	\label{eq:lem:dealer:mon:yP}
y^P(n) = p(n)+ \gamma_c n, \quad n \neq 0,
\end{equation}
is increasing by Lemma \ref{lem:schedules:hom}(a). Then using Lemma \ref{lem:adm schedule}(c), the substitution rule and an integration by parts in the form of Lemma \ref{lem:int by parts}, the monopolist dealer's goal functional can be rewritten as follows:
\begin{align*}
J_d(P) &= \int_{-\infty}^{a_P} \left(P(n^P(y))-g(y) n^P(y)-\frac{\gamma_d}{2}n^P(y)^2\right)f(y)\dd y \notag \\
&\quad+ \int_{b_P}^{\infty} \left(P(n^P(y))-g(y) n^P(y)-\frac{\gamma_d}{2}n^P(y)^2\right)f(y)\dd y \notag \\
&= \int_{-\infty}^{0} \left(P(n) -\frac{\gamma_d}{2}n^2\right)f(y^P(n)) \dd(y^P)(n) - \int_{-\infty}^{0} n h(y^P(n)) \dd(y^P)(n) \notag  \\
&\quad+\int_{0}^{\infty} \left(P(n) -\frac{\gamma_d}{2}n^2\right)f(y^P(n)) \dd(y^P)(n) - \int_{0}^{\infty} n h(y^P(n)) \dd(y^P)(n),
\end{align*}
and in turn
\begin{align}
J_d(P) &= \int_{-\infty}^{0} \left(\gamma_d n - p(n) \right)F(y^P(n)) + H(y^P(n)) \dd n  \notag \\
&\quad+\int_{0}^{\infty} \left(p(n) - \gamma_d n\right)\ol F(y^P(n)) - \ol H(y^P(n)) \dd n \notag \\
&= \int_{-\infty}^{0} \left((\gamma_d + \gamma_c) n - y^P(n) \right)F(y^P(n)) + H(y^P(n)) \dd n  \label{eq:mon:neg}\\
&\quad+\int_{0}^{\infty} \left(y^P(n) -(\gamma_d + \gamma_c) n\right)\ol F(y^P(n)) - \ol H(y^P(n)) \dd n. \label{eq:mon:pos}
\end{align}

We now establish uniqueness of the optimal price schedule. To this end, suppose that $P_*$ is an optimizer of $J_d(P)$. We seek a formula for $y^{P_*}$, which in turn yields a formula for $P_*$ via \eqref{eq:lem:dealer:mon:yP}. To this end, we employ a localised calculus of variation argument on \eqref{eq:mon:neg}--\eqref{eq:mon:pos}. We only spell this out for \eqref{eq:mon:neg}; the argument for \eqref{eq:mon:pos} is analogous. Let $K \subset (-\infty, 0) \cap \{\frac{\diff}{\diff n}y^{P_*} > 0\}$ be compact and $\kappa: \RR \setminus \{0\} \to \RR$ a continuously differentiable function that is supported on $K$ (and hence vanishes on $(0, \infty)$). Using that $y^{P_*}$ is increasing, 
	Lemma \ref{lem:nowhere dense}(b) shows that  there exists $\epsilon' > 0$ such that $y^{P_*} + \epsilon \kappa$ is increasing on $\RR \setminus \{0\}$ and hence the corresponding price schedule is admissible by Lemma \ref{lem:schedules:hom}(a) for all $\epsilon \in [-\epsilon', \epsilon']$. After plugging $y^{P_*} + \epsilon \kappa$ into \eqref{eq:mon:neg}--\eqref{eq:mon:pos}, dividing by $\varepsilon$ and sending $\varepsilon \to 0$ (and using that $\kappa$ is zero on $(0, \infty)$), optimality of $P_*$ yields
\begin{equation*}
\int_{-\infty}^{0} \left(-F(y^{P_*}(n)) + ((\gamma_d+\gamma_c) n - y^{P_*}(n) + g(y^{P_*}(n))) f(y^{P_*}(n))\right) \kappa(n) \dd n =0.
	\end{equation*}
Since $\kappa$ was arbitrary, the continuity of $\kappa$, $y^{P*}$, $f$, $g$, $F$ in turn gives
\begin{equation*}
-F(y^{P_*}(n)) + ((\gamma_d+\gamma_c) n - y^{P_*}(n) + g(y^{P_*}(n))) f(y^{P_*}(n)), \quad n \in K.
\end{equation*}
Thus, it follows that 
\begin{align*}
-F(y^{P_*}(n)) + ((\gamma_d+\gamma_c) n - y^{P_*}(n) + g(y^{P_*}(n))) f(y^{P_*}(n)) &= 0, \quad n \in (-\infty, 0) \cap \left\{\frac{\diff}{\diff n}y^{P_*} > 0\right\},\\
\ol F(y^{P_*}(n)) + ((\gamma_d+\gamma_c) n - y^{P_*}p(n) + g(y^{P_*}(n))) f(y^{P_*}(n)) &= 0, \quad n \in (0, \infty) \cap \left\{\frac{\diff}{\diff n}y^{P_*} > 0\right\}.
\end{align*}

Since $y^{P_*}$ is increasing, both $(0, \infty) \setminus \left\{\frac{\diff}{\diff n}y^{P_*} > 0\right\}$ and $(-\infty,0) \setminus \left\{\frac{\diff}{\diff n}y^{P_*} > 0\right\}$ are nowhere dense sets by Lemma \ref{lem:nowhere dense}(a). By continuity of $y^{P*}$, $f$, $g$, $F$ and $\ol F$, this implies that
\begin{align}
-F(y^{P_*}(n)) + ((\gamma_d+\gamma_c) n - y^{P_*}(n) + g(y^{P_*}(n))) f(y^{P_*}(n)) &= 0, \quad n \in (-\infty, 0), \label{eq:mon:FOC:neg}\\
\ol F(y^{P_*}(n)) + ((\gamma_d+\gamma_c) n - y^{P_*}(n) + g(y^{P_*}(n))) f(y^P(n)) &= 0, \quad n \in (0, \infty)\label{eq:mon:FOC:pos}.
\end{align}
Rearranging gives \eqref{eq:lem:dealer:mon}, so if an optimal price schedule exists it has to be of the proposed form. \\

We now verify that this price schedule is indeed optimal. To this end, note that \eqref{eq:mon:FOC:neg} together with positivity of $f$ and Assumption~\ref{ass:mon} on $F/f + \mathrm{id} - g$ imply for fixed $n \in (-\infty,0)$ that
\begin{equation}
-F(y) + ((\gamma_d+\gamma_c) n - y+ g(y)) f(y) 
\begin{cases}
> 0 &\text{if } y < y^{P_*}(n), \\
< 0 &\text{if } y> y^{P_*}(n).
\end{cases}
\label{eq:mon:derivative:neg}
\end{equation}
Similarly, \eqref{eq:mon:FOC:pos} together with positivity of $f$ and Assumption \ref{ass:mon} on $-\frac{\ol F}{f} + \mathrm{id} - g$ imply for fixed $n \in (0, \infty)$ that
\begin{equation}
\ol F(y) + ((\gamma_d+\gamma_c) n - y+ g(y)) f(y) 
\begin{cases}
> 0 &\text{if } y < y^{P_*}(n), \\
< 0 &\text{if } y>y^{P_*}(n).
\end{cases}
\label{eq:mon:derivative:pos}
\end{equation}
Now, let $P_*$ be as above and $P$ be any competitor price schedule such that $\int_{-\infty}^\infty n^P(y)^2f(y)\dd y < \infty$. Then the mean value theorem together with \eqref{eq:mon:derivative:neg} and \eqref{eq:mon:derivative:pos} implies that
\begin{align*}
J_d(P) - J_d(P_*) &= \int_{-\infty}^{0} \Big(-F(y(n)) + ((\gamma_d+\gamma_c) n - y+ g(y(n))) f(y(n)) \Big)(y^P(n)-y^{P_*}(n)) \dd n \\
&\quad+\int_{-\infty}^{0} \Big(-\ol F(y(n)) + ((\gamma_d+\gamma_c) n - y+ g(y(n))) f(y(n)\Big) (y^P(n)-y^{P_*}(n)) \dd n \\
&\leq 0 + 0 =0,
\end{align*}
where for each $n$, $y(n)$ lies in the interval with the endpoints $y^P(n)$ and $y^{P_*}(n)$. Whence, $P_*$ is indeed optimal as asserted. 
\end{proof}

\begin{proof}[Proof of Remark~\ref{rem:normal2}]

In order to prove that the optimal price schedule for the monopolist is (strictly) convex in the context of Example~\ref{ex:normal} if and only if \eqref{eq:rem:mon:beta est} holds, set 
$$y_- := (F/f+ \mathrm{id} - g)^{-1}(0) \quad \mbox{and} \quad  y_+ := (-\ol F/f+ \mathrm{id} - g)^{-1}(0).$$
Since $(F/f)'$ and $(-\bar F/f)'$ are increasing on $\RR$ and $g' = \beta$, \eqref{eq:rem:mon:est-} and \eqref{eq:rem:mon:est+} are equivalent to
\begin{align}
	(F/f)'(y_-) \leq \beta	\quad \text{and} \quad (-\bar F/f)'(y_+) \leq \beta. \label{eq:rem:mon:est:normal}
\end{align}
Define 
$$z_-:= \frac{y_- - \mu_Y}{\sigma_Y} \quad \mbox{and} \quad z_+:= \frac{y_+ - \mu_Y}{\sigma_Y}.
$$ 
Then, using the scaling properties of the normal distribution and the symmetry of $\Phi$, it follows that \eqref{eq:rem:mon:est:normal} is equivalent to 
\begin{align}
	(\Phi/\phi)'(z_-) \leq \beta	\quad \text{and} \quad (\Phi/\phi)'(-z_+) \leq \beta. \label{eq:rem:mon:est:normal:2}
\end{align}
As $(\Phi/\phi)'$ is increasing on $\RR$,  \eqref{eq:rem:mon:est:normal:2} is equivalent to 
\begin{align}
	(\Phi/\phi)'(z_{\max}) \leq \beta, \quad \mbox{where $z_{\max} := \max(z_-, -z_+)$.} \label{eq:rem:mon:est:normal:3}
\end{align}
Next, the scaling properties of the normal distribution, the symmetry of $\Phi$ and the definition of $z_-$ and $z_+$ show that $z_-$ and $z_+$ are the unique solutions of
\begin{align*}
(\Phi/\phi)(z_-) + (1- \beta)z_- = \gamma_c \frac{\mu_M}{\sigma_Y} \quad  \text{and} \quad (\Phi/\phi)(-z_+) + (1- \beta)z_+ =- \gamma_c \frac{\mu_M}{\sigma_Y}.
\end{align*}
Again using that $(\Phi/\phi)'$ is increasing on $\RR$, it follows that  $z_{\max}$ is the unique solution of 
\begin{equation}
\label{eq:rem:mon:est:normal:4}
(\Phi/\phi)(z_{\max}) + (1- \beta)z_{\max} = \gamma_c \frac{|\mu_M|}{\sigma_Y}.
\end{equation}
If $\gamma_c \frac{|\mu_M|}{\sigma_Y} \geq (\Phi/\phi)(0) = \sqrt{\pi/2}$, then it follows that $z_{\max} \geq 0$, whence \eqref{eq:rem:mon:est:normal:3} cannot be satisfied as $\beta < 1$. Conversely, if $\gamma_c \frac{|\mu_M|}{\sigma_Y} < (\Phi/\phi)(0) = \sqrt{\pi/2}$, then it follows that $z_{\max}< 0$. 

Finally, let $z < 0$ such that $(\Phi/\phi)'(z) = 1 + z (\Phi/\phi)(z) =\beta$, which exists and is unique since $(\Phi/\phi)'$ is increasing. We obtain
\begin{align}
	z \geq z_{\max}	\quad &\Leftrightarrow  \quad (\Phi/\phi)(z) + (1- \beta)z \geq \gamma_c \frac{|\mu_M|}{\sigma_Y} \notag \\
	&\Leftrightarrow  \quad 	z^2  - \frac{\gamma_c}{(1 - \beta)} \frac{|\mu_M|}{\sigma_Y} z -1 \leq  0. \quad \Leftrightarrow  \quad z \geq 	z_{\mathrm{mon}}(\beta), 	\label{eq:rem:mon:est:normal:5}
\end{align}
where 
\begin{equation*}
	z_{\mathrm{mon}}(\beta) := \frac{\gamma_c}{2 (1 -\beta)} \frac{|\mu_M|}{\sigma_Y} -  \sqrt{\frac{\gamma^2_c \mu_M^2}{4 (1 -\beta)^2\sigma_Y^2} + 1} 
\end{equation*}
denotes the negative solution of $z^2  - \frac{\gamma_c}{(1 - \beta)} \frac{|\mu_M|}{\sigma_Y} z - 1= 0$. By monotonicity of $(\Phi/\phi)'$, this implies that \eqref{eq:rem:mon:est:normal:3} is equivalent to the second part of \eqref{eq:rem:mon:beta est}
\end{proof}

\subsection{Proofs for Section~\ref{ss:oli}}

We now turn to the proofs for the Nash competition between several strategic dealers.

\begin{proof}[Proof of Lemma~\ref{lem:dealer:oli:unique}]
We prove this result under the weaker Assumptions  \ref{ass:g} and \ref{ass:mon}. Let ${\bf P_*} = (P_*, \ldots, P_*)$ be a Nash-equilibrium and $\ell^*$ and $r^*$ be defined as in \eqref{eq:lr}. Let $P_1$ be an admissible price schedule for $K$ dealers that satisfies
\begin{equation}
\label{eq:pf:lem:dealer:oli:unique:limit cond}
\lim_{n \to -\infty} P'_1(n) = \ell^* \quad \mbox{and} \quad  \lim_{n \to \infty} P'_1(n) = r^*.
\end{equation}
Set ${\bf P} = (P_1, P_*, \ldots, P_*)$. Then $\bar \ell$ and $\bar r$ defined  in \eqref{eq:ol:lr} satisfy
$\bar \ell  = \ell^*$ and $\bar r = r^*$. By Theorem \ref{thm:schedules het}, this implies that $P_1, P_*, \ldots, P_*$ are compatible. In light of the square-integrability of $f$ and the estimate \eqref{eq:lem:oli:estimate}, $K^{\bf P}(P_1)$ is always less than $\infty$ and it is greater then $-\infty$ if and only if $\int_{-\infty}^\infty n^{\bf P}_1(y)^2f(y)\dd y < \infty$. So assume in addition that $P_1$ is such that $\int_{-\infty}^\infty n^{\bf P}_1(y)^2 f(y)\dd y < \infty$. 
To ease notation, set 
$$p_* := P_*' \quad \mbox{and} \quad p_1 = P'_1.
$$
Denote the inverse function of $n^{\bf P}_1$ on $I^{\bf P}_1 = \{y \in \RR: n^{\bf P}_1(y) \neq 0\}$ by $y^{\bf P}_1$ and note that (with the convention $p_*^{-1}(x) = 0$ if $x \in [p_*(0-), p_*(0+)]$),
\begin{equation*}
	y^{\bf P}_1 (n) = p_1(n) + \gamma_c n + \gamma_c (K-1) p_*^{-1}(p_1(n)), \quad n \neq 0
\end{equation*}
is increasing and valued in $\RR \setminus \{0\}$ by Theorem \ref{thm:schedules het}(b) and (c). Now setting $H(x) = \int_{-\infty}^x h(x) dx$ and $\bar H(x) = \int_x^\infty h(x) dx$, a change of variable together with an integration by parts in the form of Lemma \ref{lem:int by parts} allows to rewrite the goal functional of dealer $1$ as
\begin{align}
K^{P_*}(P_1)  &=  \int_{\RR} \left(P_1(n^{\bf P}_1(y))-g(y) n^{\bf P}_1(y)-\frac{\gamma_d}{2}n^{\bf P}_1(y)^2\right)f(y)\dd y  \notag \\
&=\int_{\RR \setminus \{0\}} \left( \left(P_1(n) -\frac{\gamma_d}{2}n^2\right)f(y^{\bf P}_1(n)) - n h(y^{\bf P}_1(n)) \right) d y^{\bf P}_1(n)\notag \\
&= \int_{-\infty}^{0} \left(\left(\gamma_d n - p_1(n) \right)F(y^{\bf P}_1(n)) + H(y^{\bf P}_1(n))\right) \dd n \notag  \\
&\quad+\int_{0}^{\infty} \left(\left(p_1(n) - \gamma_d n\right)\ol F(y^{\bf P}_1(n)) - \ol H(y^{\bf P}_1(n))\right) \dd n \notag \\
&= \int_{-\infty}^{0} \bigg(\left(\gamma_d n - p_1(n) \right)F\Big(p_1(n) + \gamma_c n + \gamma (K-1) (p_*)^{-1}(p_1(n))\Big) \notag \\
&\qquad+  H\Big(p_1(n) + \gamma_c n + \gamma (K-1) (p_*)^{-1}(p_1(n))\Big)\bigg) \dd n \label{eq:oli:neg}  \\
&\quad+\int_{0}^{\infty} \bigg(\left(p_1(n) - \gamma_d n\right)\ol F\Big(p_1(n) + \gamma_c n + \gamma (K-1) (p_*)^{-1}(p_1(n))\Big) \notag  \\
&\qquad- \ol H\Big(p_1(n) + \gamma_c n + \gamma_c (K-1) (p_*)^{-1}(p_1(n))\Big)\bigg) \dd n. \label{eq:oli:pos}
\end{align}
Since ${\bf P_*} = (P_*, \ldots, P_*)$ is a Nash equilibrium, $P_*$ is a maximizer of $K^{P_*}(\cdot)$. In particular, it is a maximizer among all admissible price schedules that satisfy \eqref{eq:pf:lem:dealer:oli:unique:limit cond}. Now using a localised calculus of variations argument separately on \eqref{eq:oli:neg} and \eqref{eq:oli:pos} as in the proof of Lemma~\ref{lem:dealer:mon} and noting that the perturbed strategies still satisfy \eqref{eq:pf:lem:dealer:oli:unique:limit cond}, we obtain that 
\begin{align*}
-F(p_*(n) + \gamma K n) &+ \Big(\gamma_d n - p_*(n) + g(p_*(n) + \gamma K n)\Big)  \\
&\quad \times f(p_*(n) + \gamma K n) \left(1 + \frac{\gamma (K-1)}{p'_*(n)} \right) = 0, \quad n \in (-\infty,0) \cap \left\{p'_* > 0\right\}, 
\end{align*}
and
\begin{align*}
\ol F(p_*(n) + \gamma K n) &+ \Big(\gamma_d n - p_*(n) + g(p_*(n) + \gamma K n) \Big)  \\
&\quad \times f(p_*(n) + \gamma K n) \left(1 + \frac{\gamma (K-1)}{p'_*(n)} \right) = 0, \quad n \in (0, \infty) \cap \left\{p'_* > 0\right\}.
\end{align*}
Rearranging terms gives
\begin{align}
p'_*(n) &= \frac{(K-1) \gamma \left(\gamma_d n - p_*(n) + g(p_*(n) + \gamma K n)\right)}{ \frac{F}{f} \left(p_*(n) + \gamma K n \right)- (\gamma_d n - p_*(n) + g(p_*(n) + \gamma K n))},
\quad n \in (-\infty,0) \cap \left\{p'_* > 0\right\}, \label{eq:pf:lem:dealer:oli:unique:ODE:R-:1} \\
p'_*(n) &= \frac{(K-1) \gamma \left(\gamma_d n - p_*(n) + g(p_*(n) + \gamma K n)\right)}{ -\frac{\ol F}{f} \left(p_*(n) + \gamma K n \right)- (\gamma_d n - p_*(n) + g(p_*(n) + \gamma K n))},
\quad n \in (0, \infty) \cap \left\{p'_* > 0\right\}.
\label{eq:pf:lem:dealer:oli:unique:ODE:R+:1}
\end{align}
Note that the rearrangement also shows that the numerator and denominator on the right hand sides of \eqref{eq:pf:lem:dealer:oli:unique:ODE:R-:1} and \eqref{eq:pf:lem:dealer:oli:unique:ODE:R+:1} cannot be zero on $(-\infty,0)  \cap \left\{p'_* > 0\right\}$ and  $(0,\infty)  \cap \left\{p'_* > 0\right\}$, respectively.
Since $p_*$ is increasing by strict convexity of $P_*$, both $ (-\infty,0) \cap \{p_*'  > 0\}$ and $(0, \infty) \cap \{(p_*)'  > 0\}$ are nowhere dense sets by Lemma \ref{lem:nowhere dense}(a). By continuity of $F/f$, $\ol F/f$, $g$ and $p'_*$ on $\RR \setminus \{0\}$, this implies that
\begin{align}
p'_*(n) &= \frac{(K-1) \gamma \left(\gamma_d n - p_*(n) + g(p_*(n) + \gamma K n)\right)}{ \frac{F}{f} \left(p_*(n) + \gamma K n \right)- (\gamma_d n - p_*(n) + g(p_*(n) + \gamma K n))},
\quad n \in (-\infty,0), \label{eq:pf:lem:dealer:oli:unique:ODE:R-:2}  \\
p'_*(n) &= \frac{(K-1) \gamma \left(\gamma_d n - p_*(n) + g(p_*(n) + \gamma K n)\right)}{ -\frac{\ol F}{f} \left(p_*(n) + \gamma K n \right)- (\gamma_d n - p_*(n) + g(p_*(n) + \gamma K n))},
\quad n \in (0, \infty).
\label{eq:pf:lem:dealer:oli:unique:ODE:R+:2}
\end{align}
This argument also shows that the denominators on the right hand sides of \eqref{eq:pf:lem:dealer:oli:unique:ODE:R-:2} and \eqref{eq:pf:lem:dealer:oli:unique:ODE:R+:2} cannot be zero on $(-\infty,0)$ or $(0,\infty)$, respectively. Indeed, if we multiply \eqref{eq:pf:lem:dealer:oli:unique:ODE:R-:1} and \eqref{eq:pf:lem:dealer:oli:unique:ODE:R+:1} by the corresponding denominators, we get equations between two continuous functions that hold outside a nowhere dense set, hence everywhere. But this implies that if the denominator in \eqref{eq:pf:lem:dealer:oli:unique:ODE:R-:2} and \eqref{eq:pf:lem:dealer:oli:unique:ODE:R+:2} can be zero only if the numerator is. But the denominator never vanishes if the corresponding numerator does because $F/f$ and $\ol F/f$ are positive on $\RR$.
\end{proof}

Next, we establish the wellposedness results for the nonlinear ODE~\eqref{eq:lem:dealer:oli:unique:ODE} collected in Theorem~\ref{thm:ODE}. Again, we do this under weaker (but much less intuitive) assumptions than the convenient sufficient conditions imposed in Assumptions~\ref{ass:type} and \ref{ass:g:oli}.

\begin{assumption}
	\label{ass:ODE:ex}
	Suppose $f$, $g$ are continuously differentiable and there exist $\delta, C_g, C_f > 0$ such that:
	\begin{align}
		\delta &\leq 1 - g' \leq C_g < \frac{\gamma_d + K \gamma c}{K \gamma_c}, \label{eq:ass:ODE:g} \\
		\delta &\leq 1 + (F/f)' - g' \leq C_f \quad \text{on } (-\infty, y_-], \label{eq:ass:ODE:f-}\\
		\delta &\leq 1 - (\bar F/f)' - g' \leq C_f \quad \text{on } [y_+, \infty). \label{eq:ass:ODE:f+}
	\end{align}
\end{assumption}

\begin{assumption}
	\label{ass:ODE:uni}
	Suppose that $f$ and $g$ are continuously differentiable and there are $-\infty < n_- \leq 0 \leq n_+ < \infty$ and $\delta_-, \delta_+ \geq \delta $ with
	\begin{align}
		1 - g' &\geq \delta_- \text{ on } (-\infty, n_-] \quad \text{and} \quad 1 - g' \geq \delta_+  \text{ on } [n_+, \infty),
		\label{eq:ass:ODE:d limit}
	\end{align}
	such that, moreover,
	\begin{align}
		\label{eq:ass:ODE:f limit}
		\lim_{z \to -\infty} |z|^{\frac{K-1}{\delta_-^2(K + \frac{\gamma_d}{\gamma_c})}} F(z) = 0 \quad \text{and} \quad \lim_{z \to +\infty} |z|^{\frac{K-1}{\delta^2_+(K + \frac{\gamma_d}{\gamma_c})}} \bar F(z) = 0.
	\end{align}
\end{assumption}
Note that Assumptions \ref{ass:type} and \ref{ass:g:oli} from the body of the paper indeed imply Assumptions \ref{ass:ODE:ex} and \ref{ass:ODE:uni}. To wit, 
 Proposition \ref{prop:g:mon} gives \eqref{eq:ass:ODE:g} and \eqref{eq:ass:ODE:f-} and \eqref{eq:ass:ODE:f+} follows from Proposition~\ref{prop:log concave:convolution}(a) and the fact that, fo rlo-concave distributions as in Assumptions \ref{ass:type}, $f'(n) > 0$ for all sufficiently small $n$ and $f'(n) < 0$ for all sufficiently large $n$  by Proposition \ref{prop:log:concave:basics}(c). Finally, setting $\delta_+:=\delta =: \delta_-$, \eqref{eq:ass:ODE:f limit} follows from Proposition \ref{prop:log:concave:basics}(b) and an integration by parts. However, the above conditions are more general and cover, e.g., two-sided Pareto distributions with sufficiently light tails if the conditional-mean function $g$ is linear as in Example~\ref{??}.

\begin{proof}[Proof of Theorem~\ref{thm:ODE}] We prove the result under the weaker Assumptions~\ref{ass:ODE:ex} and \ref{ass:ODE:uni}. Moreover, we also prove the following two
	additional claims -- part (a) is useful for the analysis of concrete examples and part (b) will be crucial for proving Theorem \ref{thm:oli:ex}.
	\begin{enumerate}
	\item For any $\epsilon_v, \epsilon_w \in (0, 1)$ such that
	\begin{align}
		\frac{\gamma_d+ K \gamma_c}{\delta} - K \gamma_c - \frac{(K-1)\gamma_c \delta(1-\epsilon_v)}{C_f \epsilon_v} &\leq 0, \\
		(1 - \epsilon_w)\frac{\gamma_d+ K \gamma_c}{C_g} + \epsilon_w \frac{\gamma_d+ K \gamma_c}{C_f} -K \gamma_c - \frac{(K -1)\gamma_c C_g \epsilon_w}{\delta (1-\epsilon_w)}  &\geq 0,
	\end{align}
we have $\epsilon_v + \epsilon_w < 1$ and
 \begin{align}
 	P'(0-) &\in [(1-\epsilon_v) y_- + \epsilon_v(\id - g)^{-1}(0), \epsilon_w y_- + (1- \epsilon_w)(\id - g)^{-1}(0)]. \\
 	 	P'(0+) &\in [\epsilon_w y_+ + (1- \epsilon_w)(\id - g)^{-1}(0)], (1-\epsilon_v) y_+ + (\epsilon_v)(\id - g)^{-1}(0)].
 \end{align}
	\item The unique solution $P^*$ to the ODE \eqref{eq:lem:dealer:oli:unique:ODE} has derivatives that are bounded and bounded away from zero, which implies that $\lim_{n \to -\infty} P'_*(n) = -\infty$ and $\lim_{n \to \infty} P'_*(n) = \infty$. 
	\end{enumerate}

The proof is based on constructing explicit upper and lower solutions of the ODE \eqref{eq:lem:dealer:oli:unique:ODE}, and in turn use these to deduce the existence of a solution. The natural candidates for these upper and lower solutions are the functions that make the numerator and denominator in the fractions of the right-hand side of \eqref{eq:lem:dealer:oli:unique:ODE} vanish.\footnote{On $(-\infty, 0)$, the function corresponding to the numerator is the upper solution and the function corresponding to the denominator is the lower solution; on  $(0, \infty)$, the function corresponding to the numerator is the lower solution and the function corresponding to the denominator is the upper solution.} Of course, the function that makes the denominator vanish cannot really be used (since it would lead to an infinite derivative) so that another approximation argument is required. Uniqueness follows by a rather delicate Gr\"onwall estimate showing that if there were two solutions between the constructed upper and lower solutions, then the difference between them would grow so fast that at least one of them would cross the upper or lower solution, which is a contradiction. 

To ease notation, define the functions $A, B_-, B_+: \RR^2 \to \RR$ by
\begin{align*}
A(n, z) &=  (\gamma_d + K \gamma_c)n - (\id - g)(z + \gamma_c K n), \\
B_-(n, z) &= -(\gamma_d + K \gamma_c)n + \left(\id + \frac{F}{f}- g\right)(z + \gamma_c K n), \\
B_+(n, z) &= -(\gamma_d + K \gamma_c)n + \left(\id - \frac{\ol F}{f}- g\right)(z + \gamma_c K n),
\end{align*}
and set 
$$p := P'.
$$
Then, the ODE~\eqref{eq:lem:dealer:oli:unique:ODE} can be rewritten as
	\begin{align}
p'(n)  = \begin{cases} (K-1) \gamma_c  \dfrac{A(n, p(n))}{ B_-(n, p(n))},
\quad n \in (-\infty,0), \label{eq:pf:thm:ODE:R:rewritten}  \\
(K-1) \gamma_c \dfrac{ A(n, p(n))}{ B_+(n, p(n))},
\quad n \in (0, \infty).
\end{cases}
\end{align}
Note that 
\begin{equation}
B_-(n, z) = (F/f)(z + \gamma_c K n) - A(n,z) \quad \text{and} \quad B_+(n, z) = -(\ol F/f)(z + \gamma_c K n) - A(n,z). \label{eq:thm:pf:ODE:AB comparison}
\end{equation}
This implies that $A/B_-$ can only be nonnegative if $A$ is nonnegative and $A/B_+$ can only be nonnegative if $A$ is negative. Hence, if $p: \RR \setminus \{0\} \to \RR$ is increasing, continuously differentiable and satisfies \eqref{eq:pf:thm:ODE:R:rewritten} (in particular, the denominators do not vanish), then 
	\begin{align}
0 &\leq A(n, p(n)) < \frac{F}{f} \left(p(n) + \gamma_c K n \right) \quad n \in (-\infty, 0), \label{eq:thm:estimate:R-} \\
	0 &\geq A(n, p(n)) > -\frac{\ol F}{f} \left(p(n) + \gamma_c K n \right), \quad n \in (0, \infty). \label{eq:thm:estimate:R+}
	\end{align}
	Next, if a nondecreasing function $p_-$ satisfies \eqref{eq:thm:estimate:R-} and a nondecreasing function $p_+$ satisfies \eqref{eq:thm:estimate:R+}, then by the fact that $A$ is decreasing in $z$ by \eqref{eq:ass:ODE:g}, it follows that 
	\begin{equation*}
	\lim_{n \uparrow 0} p_1(n) \leq \lim_{n \downarrow 0} p_2(n).
	\end{equation*}
	Thus, the result follows if we can can show that the ODE \eqref{eq:pf:thm:ODE:R:rewritten} has a unique solution $p_-$ on $(-\infty, 0]$ whose derivatives are bounded and bounded away from zero, and a unique solution $p_+$ on $[0, \infty)$ whose derivatives are bounded and bounded away from zero. Then,
	\begin{equation}
	p(n) =
	\begin{cases}
	p_-(n) & \text{if } n < 0, \\
		p_+(n) & \text{if } n > 0,
	\end{cases}
	\end{equation}
as well as $p(0-) = p_-(0)$ and $p(0+) = p_+(0)$. 

We only establish the assertion for $p_-$, the assertion for $p_+$ follows in a similar manner.
	
	We first establish existence of a solution $p_-$ to \eqref{eq:pf:thm:ODE:R:rewritten} on $(-\infty,0)$ that has derivatives that are bounded and bounded away from zero. The idea is to construct lower and upper solutions as in Proposition~\ref{prop:walter}. Given that the right-hand side of the ODE  \eqref{eq:pf:thm:ODE:R:rewritten} on $(-\infty, 0)$ is a (multiple of) the fraction with numerator $A$ and denominator $B_-$, it is natural to consider functions $v, w: (-\infty,0] \to \RR$ such that 
	$A(n, w(n)) = 0$ and $B_-(n, v(n)) =0$. So define the functions $v, w: (-\infty,0] \to \RR$ by
	\begin{align}
	\label{eq:pf:thm:ODE:v}
	v(n) &:= \left(\id +\frac{F}{f} - g\right)^{-1} ((\gamma_d + K \gamma_c )n) - K \gamma_c n, \\
	w(n) &:= \left(\id - g\right)^{-1} ((\gamma_d + K \gamma_c )n) - K \gamma_c n.
	\label{eq:pf:thm:ODE:w}
	\end{align}
	Note that $v < w$ because $\id +F/f - g > \id - g$ by the fact that $F/f > 0$. Moreover, it follows from \eqref{eq:ass:ODE:g} and~\eqref{eq:ass:ODE:f-} that 
	$v$ and $w$ have bounded derivatives:
	
	\begin{align}
	\frac{\gamma_d+ K \gamma_c}{C_f} - K \gamma_c &\leq v'(n) \leq \frac{\gamma_d+ K \gamma_c}{\delta} - K \gamma_c, \label{eq:v prime} \\
	0 < \frac{\gamma_d+ K \gamma_c}{C_g} - K \gamma_c &\leq w'(n) \leq \frac{\gamma_d+ K \gamma_c}{\delta} - K \gamma_c. \label{eq:w prime}
	\end{align}
By definition of $v$, $w$ and \eqref{eq:thm:pf:ODE:AB comparison}, it follows that $A(n,w(n)) = 0$ and $B_-(n, w(n)) = (F/f)(w(n) + \gamma_c K n) > 0$ as well as $A(n, v(v)) = -(F/f)(v(n) + \gamma_c K n) < 0$
and $B_-(n,v(n)) = 0$. Together with \eqref{eq:pf:thm:ODE:v} and \eqref{eq:pf:thm:ODE:w}, this implies that $w$ is an upper solution of the ODE \eqref{eq:pf:thm:ODE:R:rewritten} on $(-\infty,0)$ and $v$ is \emph{essentially} a lower solution of the ODE -- note that $A(n,v(n))/B_-(n,v(n))$ is not defined but can be interpreted as $\infty$.\footnote{Indeed, one can show that $\lim_{z \downarrow v(n)} A(n,z))/B_-(n,z) = \infty$.} For this reason, we have to modify $v$ to get a proper lower solution and it will be useful to also modify $w$ to get some sharper estimates. (These refined upper and lower solutions are compared to $v$ and $w$ in Figure~\ref{fig:bounds} below.)

To this end, we first establish some estimates on the derivatives of $A$ and $B_-$ with respect to $z$. Fix $n \in (-\infty, 0]$ and let $v(n) \leq z \leq  w(n)$.
Then by \eqref{eq:ass:ODE:g} and \eqref{eq:ass:ODE:f-},\footnote{Note that since $\id - g$ is increasing, $z + K \gamma_c n \leq w(n) + K \gamma_c n = \left(\id - g\right)^{-1} ((\gamma_d + K \gamma_c )n) \leq \left(\id - g\right)^{-1}(0)$.}
	\begin{align}
	\label{eq:pf:thm:ODE:Az}
	-\delta &\geq \frac{\partial}{\partial z} A(n, z) \geq -C_g, \\
	\delta &\leq \frac{\partial}{\partial z}  B_-(n, z) \leq C_f.
	\label{eq:pf:thm:ODE:Bz}
	\end{align}
Together with the fact that $A(n, w(n)) = 0$ and $B_-(n, v(n)) =0$, this gives
	\begin{align}
	\label{eq:pf:thm:ODE:A comp}
	\delta (w(n) - z)) &\leq A(n, z) \leq C_g(w(n) - z),\\
	\delta (z(n) - v)) &\leq B_-(n, z) \leq C_f (z - v(n)).
	\label{eq:pf:thm:ODE:B comp}
	\end{align}
We proceed to construct a solution $p_{-}$ that lies strictly between $v$ and $w$. To this end, for $\epsilon \in (0, 1)$ (to be chosen sufficiently small later on) define the functions $v_\epsilon, w_\epsilon: (-\infty, 0] \to \RR$ by
	\begin{align}
	v_\epsilon(n) &:=(1 - \epsilon) v(n) + \epsilon w(n), \label{eq:veps}\\
	w_\epsilon(n) & := (1 - \epsilon) w(n) + \epsilon v(n). \label{eq:weps}
	\end{align}
Then $v < v_{\epsilon_v} < w_{\epsilon_w} < w$ for all $\epsilon_v, \epsilon_w \in (0, 1)$ with $\epsilon_v+ \epsilon_w < 1$. Moreover, for each $\epsilon \in (0, 1)$, $v_\epsilon - v = \epsilon (w - v) =  w - w_\epsilon$ and $w - v_\epsilon = (1-\epsilon)(w - v) = w_\epsilon - v$. Together with  \eqref{eq:v prime}--\eqref{eq:w prime} and \eqref{eq:pf:thm:ODE:A comp}--\eqref{eq:pf:thm:ODE:B comp}, this gives 
	\begin{align}
	v'_\epsilon - \frac{(K-1) \gamma_c A(n, v_\epsilon(n))}{B_-(n, v_\epsilon(n))} &\leq \frac{\gamma_d+ K \gamma_c}{\delta} - K \gamma_c - \frac{(K-1)\gamma_c \delta (w(n) - v_\epsilon(n))}{C_f (v_\epsilon(n) - v(n))} \notag \\
	&\leq \frac{\gamma_d+ K \gamma_c}{\delta} - K \gamma_c - \frac{(K-1)\gamma_c \delta(1-\epsilon)}{C_f \epsilon}, \label{eq:pf:thm:ODE:v'eps:estimate} \\
	w'_\epsilon -  \frac{(K-1) \gamma_c A(n, w_\epsilon(n))}{B_-(n, w_\epsilon(n))} &\geq  (1-\epsilon)\frac{\gamma_d+ K \gamma_c}{C_g} + \epsilon \frac{\gamma_d+ K \gamma_c}{C_f} -K \gamma_c \notag  \\
	&\quad- \frac{(K-1)\gamma_c C_g (w(n) - w_\epsilon(n))}{\delta (w_\epsilon(n) - v(n))} \notag \\
	&\geq (1-\epsilon)\frac{\gamma_d+ K \gamma_c}{C_g} + \epsilon \frac{\gamma_d+ K \gamma_c}{C_f} -K\gamma_c - \frac{(K - 1) \gamma_c C_g \epsilon}{\delta (1-\epsilon)}.  \label{eq:pf:thm:ODE:w'eps:estimate} 
	\end{align}
	Now, if we choose $\epsilon_v \in (0, 1)$ such that the right-hand side of \eqref{eq:pf:thm:ODE:v'eps:estimate} is nonpositive and $\epsilon_w \in (0, 1)$ such that the right-hand side of \eqref{eq:pf:thm:ODE:w'eps:estimate} is nonnegative, then we automatically have $\epsilon_v+\epsilon_w < 1$ so that $v_{\epsilon_v} < w_{\epsilon_w}$ and Proposition \ref{prop:walter} in turn shows that there exists a solution $p_-$ to the ODE \eqref{eq:lem:dealer:oli:unique:ODE} on $(-\infty,0)$ with $v_{\epsilon_v} \leq p_- \leq w_{\epsilon_w}$. In particular, we also have the additional Property (a). For normally distributed primitives, Figure~\ref{fig:bounds} illustrates how the refined upper and lower solutions $w_{\epsilon_w}$, $v_{\epsilon_v}$ improve the bounds that can be gleaned from $w$, $v$.
	
\begin{figure}
\begin{center}
\includegraphics[width=0.45\textwidth]{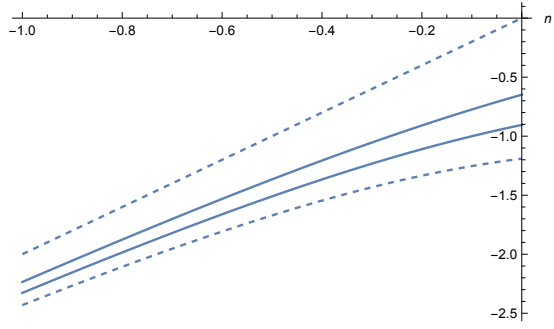}	
\caption{The functions $w$, $v$ and the upper and lower solutions $w_{\epsilon_w}$, $v_{\epsilon_v}$ for $K=2$ dealers, standard normal noise, client signals and inventories, and inventory costs $\gamma_c=1$, $\gamma_d=0$.}
\label{fig:bounds}
\end{center}
\end{figure}

Moreover, Property (b) follows from \eqref{eq:pf:thm:ODE:R:rewritten} and \eqref{eq:pf:thm:ODE:A comp}--\eqref{eq:pf:thm:ODE:B comp} via
	\begin{align*}
p_-'(n) &= \frac{(K-1) \gamma_c A(n, p(n))}{B_-(n, p(n))} \geq \frac{ (K-1) \gamma_c \delta (w(n) - p(n))}{C_f (p(n) - v(n))} \\
&\geq \frac{ (K-1) \gamma_c \delta (w(n) - w_{\epsilon_w}(n))}{C_f (w_{\epsilon_w}(n) - v(n))} \geq	\frac{ (K-1) \gamma_c \delta {\epsilon_w}}{C_f (1- {\epsilon_w})}, \\ 
p_-'(n) &= \frac{(K-1) \gamma_c A(n, p(n))}{B_-(n, p(n))} \leq \frac{(K-1) \gamma_c C_g (w(n) - p(n))}{\delta (p(n) - v(n))} \\
&\leq \frac{(K-1) \gamma_c C_g (w(n) - v_{\epsilon_v}(n))}{\delta (v_{\epsilon_v}(n) - v(n))} \leq \frac{(K-1) \gamma_c C_g (1-{\epsilon_v})}{\delta {\epsilon_v}}.
	\end{align*}
		Finally, we establish uniqueness of a solution $p_-$ to \eqref{eq:lem:dealer:oli:unique:ODE} that satisfies \eqref{eq:thm:estimate:R-}. It follows from 
		\eqref{eq:ass:ODE:d limit} that for $n \leq n_-$,
	\begin{align}
					\label{eq:deriv:w small}
							w'(n) + K \gamma_c &\leq \frac{\gamma_d+ K \gamma_c}{\delta_-}, \\
				\label{eq:deriv:Az small}
				\frac{\partial}{\partial z}A(n, z) &\leq -\delta_-.
\end{align}
Seeking a contradiction, suppose there are two solutions $z_1, z_2$ to \eqref{eq:lem:dealer:oli:unique:ODE} that satisfy \eqref{eq:thm:estimate:R-}. By local uniqueness (the right-hand side of \eqref{eq:lem:dealer:oli:unique:ODE} is local Lipschitz-continuous whenever it is well defined), it follows that $z_1$ and $z_2$ are ordered everywhere. Hence, we may assume without loss of generality that $w \geq z_1 > z_2 > v$, where the last inequality follows from \eqref{eq:thm:estimate:R-}. Set $\Delta z :=  z_1 - z_2$. Using the growth conditions of $g$ and $F$ in \eqref{eq:ass:ODE:d limit} and \eqref{eq:ass:ODE:f limit}, we aim to show that then $\Delta z(n) >  w(n) - v(n)$ for $n$ sufficiently small, which yields a contradiction. Using that $A(n, z)$ is decreasing in $z$ by \eqref{eq:pf:thm:ODE:Az} and $B_-(n, z)$ is increasing in $z$ by \eqref{eq:pf:thm:ODE:Bz}, it follows from \eqref{eq:deriv:Az small} and \eqref{eq:thm:pf:ODE:AB comparison} (recalling that $A(n, w(n)) = 0$) that for, $n \leq n_-$,
	\begin{align*}
	\Delta z'(n) &= (K-1)\gamma_c \left(\frac{A(n, z_1(n))}{B_-(n, z_1(n))} - \frac{A(n, z_2(n))}{B_-(n, z_2(n))}\right) \leq (K-1)\gamma_c  \frac{A(n, z_1(n))-A(n, z_2(n))}{B_-(n, z_1(n))} \\
	&\leq -\frac{(K-1)\gamma_c  \delta_- \Delta z(n)}{B_-(n, w(n))} = - \frac{(K-1)\gamma_c \delta_-\delta z(n)}{\frac{F}{f}(w(n)+K \gamma_c n)}.
	\end{align*}
Using \eqref{eq:deriv:w small} and Gr\"onwall's lemma, we obtain
\begin{align*}
\Delta z(n) &\geq \Delta z(n_-)\exp\left((K-1)\gamma_c \delta_- \int_n^{n-} \frac{f}{F}\Big(w(m) + K \gamma_c m\Big) \dd m\right) \\
&\geq \Delta z(n_-)\exp\left(\frac{(K-1)\gamma_c}{\gamma_d + K \gamma_c} \delta_-^2 \int_n^{n-} \frac{f}{F}\Big(w(m) + K \gamma_c m\Big) (w'(m)  + K \gamma_c)\dd m\right) \\
&= \Delta z(n_-)\exp\left(\frac{(K-1)\gamma_c}{\gamma_d + K \gamma_c} \delta_-^2 \Big(\log\big(F(w(n_-)+ K \gamma_c n_-)\big)  - \log\big(F(w(n)+ K \gamma_c n)\big)  \Big)\right) \\
&= \Delta z(n_-) \big(F(w(n_-)+ K \gamma_c n_-)\big)^{\frac{(K-1)\gamma_c}{\gamma_d + K \gamma_c} \delta_-^2 } \big(F(w(n)+ K \gamma_c n)\big)^{-\frac{(K-1)\gamma_c}{\gamma_d + K \gamma_c} \delta_-^2 }. 
\end{align*}
We arrive at the contradiction $\Delta z(n) > w(n) - v(n)$ for $n$ sufficiently small, if we can show that
\begin{equation}
\label{eq:pf:thm:ODE:limit}
\lim_{n \to \infty} (w(n) - v(n))\big(F(w(n)+ K \gamma_c n)\big)^{\frac{(K-1)\gamma_c}{\gamma_d + K \gamma_c} \delta_-^2} =0.
\end{equation}
To this end, note that by the fact that $\lim_{n \to \infty} w(n) + K \gamma_c n = -\infty$ by \eqref{eq:w prime}, de l'H\^opital, \eqref{eq:v prime} and \eqref{eq:w prime}, 
\begin{equation*}
\limsup_{n \to -\infty} \frac{w(n) - v(n)}{|w(n) + K \gamma_c n|} \leq \frac{\limsup_{n \to -\infty} (v'(n) - w'(n))}{\liminf_{n \to \infty} w'(n) + K \gamma_c} \leq \frac{C_g}{\delta} -1 < \infty.
\end{equation*}
Moreover, by \eqref{eq:ass:ODE:f limit} and the fact that $\lim_{n \to \infty} w(n) + K \gamma_c n = -\infty$, we obtain
\begin{equation*}
\lim_{n \to \infty} |w(n) + K \gamma_c n| \big(F(w(n)+ K \gamma_c n)\big)^{\frac{(K-1)\gamma_c}{\gamma_d + K \gamma_c} \delta_-^2} =0.
\end{equation*}
Combining these two limits gives \eqref{eq:pf:thm:ODE:limit} and thereby completes the proof.
\end{proof}

\begin{remark}\label{rem:recipe}
The upper and lower solutions constructed in the proof of Theorem~\ref{thm:ODE}] can be used to solve the ODE~\ref{eq:lem:dealer:oli:unique:ODE} numerically as follows:
\begin{enumerate}
\item[(a)] Choose the optimal values of the constants $\varepsilon_v$, $\varepsilon_w$ by solving the quadratic equations obtained by setting the right-hand sides of~\eqref{eq:pf:thm:ODE:v'eps:estimate}  and \eqref{eq:pf:thm:ODE:w'eps:estimate}  to zero.
\item[(b)] With these values of $\varepsilon_v$, $\varepsilon_w$ and the explicit functions $v$, $w$ from \eqref{eq:pf:thm:ODE:v}, \eqref{eq:pf:thm:ODE:w}, the functions $v_{\varepsilon_v}$ and $w_{\varepsilon_w}$ from \eqref{eq:veps}, \eqref{eq:weps} are given in closed form and in turn provide upper and lower for the exact solution of the ODE~\eqref{eq:lem:dealer:oli:unique:ODE}.
\item[(c)] Starting from these upper and lower bounds at some negative and positive values $n_{-}$ and $n_+$, solve the~\ref{eq:lem:dealer:oli:unique:ODE} on $[n_{-},0]$ and $[0,n_+]$ with a standard  ODE solver for uniformly Lipschitz ODEs. This in turn leads to upper and lower bounds for the exact solution, as depicted in Figure~\ref{fig:upperlower}. Already for moderate values of $n_{-}$, $n_{+}$, these upper and lower solutions converge very quickly. They therefore provide extremely accurate bounds for the exact solution and, in particular, its value at $0-$ and $0+$ that are crucial for the application of the Verification Theorem~\ref{thm:oli:ex}.
\end{enumerate}
\end{remark}

Finally, we prove the Verification Theorem~\ref{thm:oli:ex}, which ensures that solution to the  ODE~\eqref{eq:lem:dealer:oli:unique:ODE} indeed identifies a Nash equilibrium. 

\begin{proof}[Proof of Theorem~\ref{thm:oli:ex}] We prove the result under the weaker Assumptions \ref{ass:mon},
			\ref{ass:ODE:ex}, \ref{ass:ODE:uni} and \ref{ass:f:oli}. The idea of the proof is to show by a direct argument that given the candidate price schedule $P_*$ for dealers $k = 2, \ldots, K$, any deviation for dealer $k = 1$ from the candidate $P_*$ is suboptimal, i.e., $K^{P_*}(P_1)  \leq K^{P_*}(P^*)$. To this end, we write $K^{P_*}(P_1) = \int_{\RR} \eta (n, (P'_*)^{-1}(P'_1(n)), P'_1(n))  \dd n$ for a suitable function $\eta: \RR^3 \to \RR$ and establish the pointwise optimality
$$\eta \Big(n, (P'_*)^{-1}(P'_1(n)), P'_1(n)\Big) \leq \eta \Big(n, (P'_*)^{-1}(P'_*(n)), P'_*(n)\Big), \quad n \neq 0.$$
 Given the bid-ask spread at $n=0$ zero, this is rather delicate. 
A similar sufficient optimality condition also appears in \cite[Equation (10)]{back.baruch.13}, but is only verified in a number of concrete examples, e.g. normally-distributed client types.
		
Let $P_1$ be an admissible price schedule for $K$ dealers and set
\begin{equation}
\lim_{n \to -\infty} P'_1(n) =: \ell_1 \quad \mbox{and} \quad  \lim_{n \to \infty} P'_1(n) =: r_1.
\end{equation}
Since $\lim_{n \to -\infty} P'_*(n) = -\infty$ and $\lim_{n \to \infty} P'_*(n) = \infty$ by the proof of Theorem \ref{thm:ODE}, $\bar \ell$ and $\bar r$ from \eqref{eq:ol:lr} satisfy $\bar \ell  = \ell_1$ and $\bar r = r_1$. Hence, $P_1$ is automatically compatible with $P_*$ by Theorem \ref{thm:schedules het}. Set ${\bf P} = (P_1, P_*, \ldots, P_*)$ and 
$$p_* = P'_*, \quad p_1 = P'_1.
$$
Now setting $H(x) = \int_{-\infty}^x h(x) dx$ and $\bar H(x) = \int_x^\infty h(x) dx$, the same calculations as in \eqref{eq:oli:pos} give
\begin{align}
K^{P_*}(P_1) &= \int_{-\infty}^{0} \eta_-\Big(n, (p_*)^{-1}(p_1(n)), p_1(n)\Big) \dd n + \int_{-\infty}^{0} \eta_+\Big(n, (p_*)^{-1}(p_1(n)), p_1(n)\Big) \dd n,
\end{align}
where $\eta_-, \eta_+ : \RR^3 \to \RR$ are given by
\begin{align}
\label{eq:pf:thm:oli:ex:eta-}
\eta_-(n, x, z) &= (\gamma_d n - z)F(z + \gamma_c n + \gamma_c(K-1) x) + H(z + \gamma_c n + \gamma_c(K-1) x), \\
\eta_+(n, x, z) &= -(\gamma_d n - z)\ol F(z + \gamma_c n + \gamma_c(K-1) x) - \ol H(z + \gamma_c n + \gamma_c(K-1) x).
\label{eq:pf:thm:oli:ex:eta+}
\end{align}
To establish optimality of $P_*$, it suffices to establish pointwise optimality, that is,
\begin{align}
\eta_-\left(n, n, p_*(n)\right) &\geq  \eta_-\left(n, (p_*)^{-1}(p_1(n)),p_1(n)\right) , \quad n \in (-\infty, 0), 
\label{eq:pf:thm:oli:ex:eta- ineq}\\
\eta_+\left(n, n, p_*(n)\right) &\geq  \eta_+\left(n, (p_*)^{-1}(p_1(n)),p_1(n)\right) , \quad n \in (0, \infty). \label{eq:pf:thm:oli:ex:eta+ ineq}
\end{align}
We only establish \eqref{eq:pf:thm:oli:ex:eta- ineq}; \eqref{eq:pf:thm:oli:ex:eta+ ineq} follows by a similar argument.

We first derive some preliminary estimates on derivatives of the function $\eta_-$. To this end, define the functions  $A, B_-, B_+: \RR^3 \to \RR$ by\footnote{In view of the proof of Theorem \ref{thm:ODE}, this is a slight abuse of notation. But this is justified as $A(n, x, z)$ coincides with $A(n, z)$ from Theorem \ref{thm:ODE} for $x = n$, and the same is true for $B_-$ and $B_+$.}
\begin{align*}
A(n, x, z) &=(\gamma_d + \gamma_c) n + (K-1)x - (\id - g)(z + \gamma_c n + \gamma_c (K -1) x), \\
B_-(n, x, z) &= -(\gamma_d + \gamma_c) n- (K-1)x + \left(\id + \frac{F}{f} - g\right)(z + \gamma_c n + \gamma_c (K -1) x), \\
B_+(n, x, z) &= -(\gamma_d + \gamma_c) n- (K-1)x + \left(\id - \frac{\ol F}{f} - g\right)(z + \gamma_c n + \gamma_c (K -1) x).
\end{align*}
It follows from \eqref{eq:ass:ODE:g}, \eqref{eq:ass:ex:f:f-} and \eqref{eq:ass:ex:f:f+} that
\begin{align}
\label{eq:pf:thm:oli:ex:dA dn}
\frac{\partial}{\partial n} A(n, x, z)  &\geq (\gamma_d + \gamma_c) - \gamma_c \left(1 + \frac{\gamma_d}{K \gamma_c} \right) \geq 0, \\
\frac{\partial}{\partial n} B_-(n, x, z)  &\leq -(\gamma_d + \gamma_c) + \gamma_c \left(1 + \frac{\gamma_d}{\gamma_c} \right) = 0,  \quad \text{ if } z + \gamma_c n + \gamma_c (K -1) x \leq p^*(0-),
\label{eq:pf:thm:oli:ex:dB- dn} \\
\label{eq:pf:thm:oli:ex:dB+ dn}
\frac{\partial}{\partial n} B_+(n, x, z)  &\leq -(\gamma_d + \gamma_c) + \gamma_c \left(1 + \frac{\gamma_d}{\gamma_c} \right) = 0,  \quad \text{ if } z + \gamma_c n + \gamma_c (K -1) x \geq p^*(0+).
\end{align}
Also note that
\begin{equation}
\label{eq:pf:thm:oli:ex:B ineq}
B_+(n, x, z) = B_-(n, x, z) - \frac{1}{f}(z + \gamma_c n + \gamma_c (K -1) x) \leq B_-(n, x, z).
\end{equation}
The above implies that
\begin{equation}
\label{eq:pf:thm:oli:ex:B 0 0}
B_-(n, 0, z) \geq B_-(0, 0, y_-) = 0 \quad  \text{if }  z \geq p_*(0-).
\end{equation}
Indeed, the equality in \eqref{eq:pf:thm:oli:ex:B 0 0} follows from the fact that $y_- = (\id + F/f - g)^{-1}(0) $. For the inequality in \eqref{eq:pf:thm:oli:ex:B 0 0}, recall that $y_- \leq p_*(0-)$ by Theorem \ref{thm:ODE}. We distinguish two cases: First, if $z+ \gamma_c n \leq y_- \leq p_*(0-)$, \eqref{eq:pf:thm:oli:ex:dB- dn} and positivity of $\id + F/f - g$ on $[y_-, \infty)$ give 
\begin{equation*}
	B_-(n, 0, z) \geq B_-(0, 0, z) \geq B_-(0, 0, y_-).
\end{equation*}
Next, if $z+ \gamma_c n > y_-$, there is $y_- < z' < z$ with $z'+ \gamma_c n= y_-$. Then by positivity of $\id + F/f - g$ on $[y_-, \infty)$, $B_-(n, 0, z) \geq B_-(n, 0, z')$, and for $z'$ the inequality follows as in the first case.

The importance of $A, B_-, B_+$ becomes clear when we note that the ODE \eqref{eq:lem:dealer:oli:unique:ODE} can be written as
\begin{equation}
p'_*(x) =
\begin{cases}
\dfrac{(K-1) \gamma_c A(x, x, p_*(x))}{B_-(x, x,  p_*(x))} & \text{ if } x < 0, \\
\dfrac{(K-1) \gamma_c A(x, x, p_*(x))}{B_+(x, x,  p_*(x))} & \text{ if } x > 0,
\end{cases}
\label{eq:pf:thm:oli:ex:pstar prime}
\end{equation}
and by the definition of $\eta_-$ in \eqref{eq:pf:thm:oli:ex:eta-},
\begin{align}
\frac{\partial}{\partial x} \eta_-(n, x, z) &= (K-1) \gamma_c f\left(z + \gamma_c n + \gamma_c (K -1) x\right) A(n, x, z), 
\label{eq:pf:thm:oli:ex:deta- dx}
\\
\frac{\partial}{\partial z} \eta_-(n, x, z) &= - f\left(z + \gamma_c n + \gamma_c (K -1) x\right) B_-(n, x, z).
\label{eq:pf:thm:oli:ex:deta- dz}
\end{align}

After these preparations fix $n \in (-\infty, 0)$ and set $x := (p_*)^{-1}(p_1(n))$ and $z := p_1(n)$. We shall distinguish the two cases $x \leq n$ and $x > n$. For the latter, we have to consider the subcases $x \in (n, 0)$, $x = 0$ and $x > 0$. This is due due to the fact that price schedules are discontinuous at zero.

{\bf Case 1.} Let $x \leq n$. Then $z = p_*(x)$. By \eqref{eq:pf:thm:oli:ex:dA dn} and \eqref{eq:pf:thm:oli:ex:dB- dn}, we obtain for $\xi \leq n$
\begin{align*}
\frac{\partial}{\partial x} \eta_-(n, \xi, p^*(\xi)) &= (K-1) \gamma_c f\left(p^*(\xi) + \gamma_c n + \gamma_c (K -1) \xi \right) A(n, \xi, p^*(\xi)) \notag \\
&\geq (K-1) \gamma_c f\left(p^*(\xi) + \gamma_c n + \gamma_c (K -1) \xi\right) A(\xi, \xi, p^*(\xi)),
\\
\frac{\partial}{\partial z} \eta_-(n, x, p^*(\xi)) &= -f\left(p^*(\xi) + \gamma_c n + \gamma_c (K -1) \xi \right) B_-(n, \xi, p^*(\xi)) \notag \\
&\geq -f\left(p^*(\xi) + \gamma_c n + \gamma_c (K -1) \xi \right) B_-(\xi, \xi, p^*(\xi)).
\end{align*}
Combining this with the ODE~\eqref{eq:pf:thm:oli:ex:pstar prime} for $p_*$, we obtain
\begin{equation*}
\frac{\diff \eta_-}{\diff x}  (n,\xi, p_*(\xi)) +  p'_*(\xi) \frac{\diff \eta_-}{\diff z} (n,\xi, p_*(\xi)) \geq 0, \quad \xi \leq n.
\end{equation*}
We conclude that
\begin{align*}
\eta_-(n, x, p_*(x)) &= \eta_-(n, n, p_*(n) - \int_x^n \left( \frac{\diff}{\diff x}  \eta_-(n,\xi, p_*(\xi)) +  p'_*(\xi) \frac{\diff}{\diff z}  \eta_-(n,\xi, p_*(\xi))\right) \dd \xi \\
&\leq \eta_-(n, n, p_*(n)).
\end{align*}

{\bf Case 2(a).} Let $n < x < 0$. Then $z = p_*(x)$, and a similar argument as in Case 1 gives 
\begin{equation*}
\frac{\diff \eta_-}{\diff x}  (n,\xi, p_*(\xi)) +  p'_*(\xi) \frac{\diff \eta_-}{\diff z} (n,\xi, p_*(\xi)) \leq 0, \quad n < \xi \leq n,
\end{equation*}
and hence 
\begin{align}
\eta_-(n, x, p_*(x)) &= \eta_-(n, n, p_*(n) + \int_x^n \left( \frac{\diff}{\diff x}  \eta_-(n,\xi, p_*(\xi)) +  p'_*(\xi) \frac{\diff}{\diff z}  \eta_-(n,\xi, p_*(\xi))\right) \dd \xi  \notag \\
&\leq \eta_-(n, n, p_*(n)).
\label{eq:pf:thm:oli:case 2}
\end{align}

{\bf Case 2(b).} Let $x = 0$. Then $p_*(0-) \leq z \leq p_*(0+)$, and taking limits in \eqref{eq:pf:thm:oli:case 2} for $x \uparrow 0$ gives
\begin{equation*}
\eta_-(n, n, p_*(n)) \geq \eta_-(n, 0, p_*(0-)).
\end{equation*}
Next, it follows from \eqref{eq:pf:thm:oli:ex:deta- dz} and \eqref{eq:pf:thm:oli:ex:B 0 0} that
\begin{equation*}
\eta_-(n, 0, p_*(0-)) \geq \eta_-(n, 0, z).
\end{equation*}
Combining these two estimates in turn gives
\begin{equation}
\eta_-(n, n, p_*(n)) \geq \eta_-(n, 0, z).
\label{eq:pf:thm:oli:case 3}
\end{equation}

{\bf Case 2(c).} Let $x > 0$. Then $p_*(x) =  z$ and \eqref{eq:pf:thm:oli:case 3} for $z = p(0+)$ give
\begin{equation*}
\eta_-(n, n, p_*(n)) \geq \eta_-(n, 0, p(0+)).
\end{equation*}
Moreover, by \eqref{eq:pf:thm:oli:ex:dA dn}, \eqref{eq:pf:thm:oli:ex:B ineq} and \eqref{eq:pf:thm:oli:ex:dB+ dn}, we obtain for $\xi > 0$,
\begin{align*}
\frac{\partial}{\partial x} \eta_-(n, \xi, p^*(\xi)) &\leq (K-1) \gamma_c f\left(p^*(\xi) + \gamma_c n + \gamma_c (K -1) \xi\right) A(\xi, \xi, p^*(\xi)),
\\
\frac{\partial}{\partial z} \eta_-(n, x, p^*(\xi)) &= -f\left(p^*(\xi) + \gamma_c n + \gamma_c (K -1) \xi \right) B_-(n, \xi, p^*(\xi)) \notag \\
&\leq -f\left(p^*(\xi) + \gamma_c n + \gamma_c (K -1) \xi \right) B_+(n, \xi, p^*(\xi)) \notag \\
&\leq -f\left(p^*(\xi) + \gamma_c n + \gamma_c (K -1) \xi \right) B_+(\xi, \xi, p^*(\xi)).
\end{align*}
Combining this with \eqref{eq:pf:thm:oli:ex:pstar prime} gives
\begin{equation*}
\frac{\diff \eta_-}{\diff x}  (n,\xi, p_*(\xi)) +  p'_*(\xi) \frac{\diff \eta_-}{\diff z} (n,\xi, p_*(\xi)) \leq 0, \quad \xi > 0.
\end{equation*}
We conclude that
\begin{align*}
\eta_-(n, x, p_*(x)) &= \eta_-(n, 0, p_*(0+) + \int_0^x \left( \frac{\diff}{\diff x}  \eta_-(n,\xi, p_*(\xi)) +  p'_*(\xi) \frac{\diff}{\diff z}  \eta_-(n,\xi, p_*(\xi))\right) \dd \xi \\
&\leq \eta_-(n, 0, p_*(0+)).
\end{align*}
Combining the above estimates in turn gives
\begin{equation*}
	\eta(n, n, p_*(n)) \geq \eta_-(n, x, p_*(x)).
\end{equation*}
Putting everything together establishes pointwise optimality  in \eqref{eq:pf:thm:oli:ex:eta- ineq} and thereby completes the proof.
\end{proof}

\begin{proof}[Proof of Remark~\ref{rem:normaoli}]
Using the notation of Example \ref{ex:normal}, define $h_-, h_+ : \RR \to \RR$ by
	\begin{align*}
h_-(y) &:= \beta \sigma_Y \frac{f}{F}(y) + y  - \mu_Y - \gamma_c \mu_M, \\
h_+(y) &:= -\beta \sigma_Y \frac{f}{\ol F}(y) + y  - \mu_Y - \gamma_c \mu_M.
\end{align*}
Set $\tilde y_- :=  \inf\{y \in \RR: h_-(y) = 0\}$ and $\tilde y_+ :=  \inf\{y \in \RR: h_+(y) = 0\}$. We proceed to show that under condition \eqref{eq:rem:dealer:oli:norm:est}, 
\begin{align}
0 <	h'_- &\leq 1 \quad  \text{on }   (-\infty, \tilde y-], &\text{and} && 0 < h'_+ &\leq 1 \quad  \text{on }   [\tilde y+,\infty), \label{eq:rem:dealer:oli:norm:est:h-} \\
	 (F/f)' - \beta &\leq 0 \quad \text{on } (-\infty, \tilde y_-] & \text{and} && (\bar F/f)' - \beta &\leq 0\quad \text{on } [\tilde y_+, \infty). \label{eq:rem:dealer:oli:norm:est:equil} 
\end{align}
We only establish the first parts of \eqref{eq:rem:dealer:oli:norm:est:h-} and \eqref{eq:rem:dealer:oli:norm:est:equil}. The proof for the second parts are analogous. Set $\tilde z_- := \tfrac{\tilde y_- -\mu_Y}{\sigma Y}$. Then the scaling properties and the symmetry of the normal distribution imply that the first parts of \eqref{eq:rem:dealer:oli:norm:est:h-} and \eqref{eq:rem:dealer:oli:norm:est:equil} are equivalent to 
\begin{align}
0 \geq  \beta \left(\frac{\phi}{\Phi}\right)'(z) &> -1 \quad \text{on }   (-\infty, \tilde z-], \label{eq:rem:dealer:oli:norm:est:h-:2} \\
1 + z\frac{\Phi}{\phi}(z)  &\leq \beta \quad \text{on }   (-\infty, \tilde z-]. \label{eq:rem:dealer:oli:norm:est:equil:2} 
\end{align}
Since $\beta \in (0,1)$ and $(\frac{\phi}{\Phi})' \in (0, 1)$ on $\RR$, \eqref{eq:rem:dealer:oli:norm:est:h-:2} is automatically satisfied, and since $(\frac{\Phi}{\phi})'$ is increasing on $\RR$, \eqref{eq:rem:dealer:oli:norm:est:equil:2} is equivalent to 
\begin{align}
 1 + \tilde z_-\frac{\Phi}{\phi}(\tilde z_-)  \leq \beta. \label{eq:rem:dealer:oli:norm:est:equil:3} 
\end{align}
To establish \eqref{eq:rem:dealer:oli:norm:est:equil:3}, note that the second part of \eqref{eq:rem:dealer:oli:norm:est} together with the fact that $\frac{\gamma_c |\mu_M|}{\sigma_Y} \leq 1$ by the first part of \eqref{eq:rem:dealer:oli:norm:est} and the definition of $z_{\mathrm{oli}}(\beta)$ yield
\begin{align*}
\beta\frac{\phi}{\Phi}\left(z_{\mathrm{oli}}(\beta)\right) + z_{\mathrm{oli}}(\beta) -\frac{\gamma_c \mu_M}{\sigma_Y} &\geq \frac{\phi}{\Phi}\left(z_{\mathrm{oli}}(\beta)\right) \left(\beta + \frac{\Phi}{\phi}(z_{\mathrm{oli}}(\beta))\left(z_{\mathrm{oli}}(\beta) -\frac{\gamma_c |\mu_M|}{\sigma_Y} \right)\right) \\
&\geq \frac{\phi}{\Phi}\left(z_{\mathrm{oli}}(\beta)\right)\left(\beta + \frac{2 \beta - 1}{\frac{\gamma_c |\mu_M|}{\sigma_Y} } \left(z_{\mathrm{oli}}(\beta) -\frac{\gamma_c |\mu_M|}{\sigma_Y} \right)\right) \\
&\geq \frac{\phi}{\Phi}\left(z_{\mathrm{oli}}(\beta)\right)\left(\beta - \beta c \right) \geq 0 =\beta\frac{\phi}{\Phi}\left(\tilde z_-\right) + \tilde z_- -\frac{\gamma_c \mu_M}{\sigma_Y}.
\end{align*}
Hence, $z_{\mathrm{oli}}(\beta) \geq \tilde z_-$ by the definition of $\tilde z_-$. Next, using that $\frac{\Phi}{\phi}(\tilde z_-) \geq \frac{\beta}{\frac{\gamma_c |\mu_M|}{\sigma_Y} - \tilde z_-}$ by definition of $\tilde z_-$ and using that $\tilde z_- \leq z_{\mathrm{oli}}(\beta) \leq 0$ gives
\begin{align*}
	1 + \tilde z_-\frac{\Phi}{\phi}(\tilde z_-)  \leq 1 + \frac{\beta \tilde z_-}{\frac{\gamma_c |\mu_M|}{\sigma_Y} - \tilde z_-} \leq 1 -\beta +   \frac{\beta\frac{\gamma_c |\mu_M|}{\sigma_Y}}{\frac{\gamma_c |\mu_M|}{\sigma_Y} - \tilde z_-} \leq  1 -\beta +   \frac{\beta\frac{\gamma_c |\mu_M|}{\sigma_Y}}{\frac{\gamma_c |\mu_M|}{\sigma_Y} - z_{\mathrm{oli}}(\beta)} = \beta,
\end{align*}
and we have \eqref{eq:rem:dealer:oli:norm:est:equil:3}. Next, define the function
 $u: \RR \setminus \{0\} \to \RR$ by
	\begin{equation}
	\label{eq:rem:dealer:oli:norm}
u(n) := \begin{cases}
h_-^{-1}\Big( \gamma_c Kn\Big) - \gamma_c K n, &\text{ if } n < 0, \\
h_+^{-1}\Big( \gamma_c Kn\Big) - \gamma_c K n &\text{ if } n > 0.
	\end{cases}
\end{equation}
Then $u$ is continuously differentiable and nondecreasing by \eqref{eq:rem:dealer:oli:norm:est:h-}. Moreover, it satisfies the ODE
\begin{align}
u'(n) = \begin{cases}
\dfrac{K \gamma_c \left(- u(n) + g(u(n) + \gamma_c K n)\right)}{ \frac{F}{f} \left(u(n) + \gamma_c K n \right)- (- u(n) + g(u(n) + \gamma_c K n))}, &\text{ if  }n \in (-\infty,0), 
\\
\dfrac{\gamma_c K \left( - u(n) + g(u(n) + \gamma_c K n)\right)}{ -\frac{\ol F}{f} \left(u(n) + \gamma_c K n \right)- ( -u(n) + g(u(n) + \gamma_c K n))}, &\text{ if  } n \in (0,\infty). \end{cases}
\label{eq:rem:dealer:oli:norm:ODE-}
\end{align}
We only establish \eqref{eq:rem:dealer:oli:norm:ODE-} on $(-\infty, 0)$. To this end, fix $n < 0$ and set $x := u(n) + \gamma_c K n$. Using the definition of $u$, the identity $u(n) + \gamma_c K n = x =  h_-^{-1}( \gamma_c K n)$, the formula $h_-'(x) = -\beta (x -\mu_Y) \frac{f}{F} -\beta \sigma_Y \frac{f^2}{F^2} +1$ and the identity $(-h_- + \mathrm{id} - g)(x) =  -\beta \sigma_Y \frac{f}{F}(x) - \beta (x- \mu_Y)$, we obtain
\begin{align*}
	u'(n)  &= \frac{\gamma_c K} {h_-'(x)} - \gamma_c K n =  \frac{\gamma_c K} {-\beta (x -\mu_Y) \frac{f}{F} -\beta \sigma_Y \frac{f^2}{F^2} +1} - \gamma_c K n \\
	 &=  \frac{\gamma_c K \frac{F}{f}(x)} {-\beta (x -\mu_Y) -\beta \sigma_Y \frac{f}{F} +\frac{F}{f}(x) } - \gamma_c K n =\frac{K \gamma_c \frac{F}{f}(x)}{ - h_-(x) +  (\mathrm{id} - g)(x) + \frac{F}{f}(x) } - \gamma_c K n \\
&= \frac{K \gamma_c \left(h_-(x) - (\mathrm{id} - g)(x)\right)}{ - h_-(x) + \frac{F}{f}(x) +  (\mathrm{id} - g)(x)} = 	\frac{K \gamma_c \left(\gamma_c K n - (\mathrm{id} - g)(x)\right)}{ - \gamma_c K n  + \frac{F}{f}(x) +  (\mathrm{id} - g)(x)} \\
&= \frac{K \gamma_c \left(- u(n) + g(u(n) + \gamma_c K n)\right)}{ \frac{F}{f} \left(u(n) + \gamma_c K n \right)- (- u(n) + g(u(n) + \gamma_c K n))}, \quad  n \in (-\infty, 0).
\end{align*}

Finally, since $u$ is nondecreasing and satisfies the ODE \eqref{eq:rem:dealer:oli:norm:ODE-}, it follows that
\begin{align*}
&u'(n) - \frac{(K-1) \gamma_c \left(- u(n) + g(u(n) + \gamma_c K n)\right)}{ \frac{F}{f} \left(u(n) + \gamma_c K n \right)- (- u(n) + g(u(n) + \gamma_c K n))} = \frac{1}{K} u'(n) \geq 0, & n \in (-\infty,0),
\\
&u'(n) - \frac{\gamma_c (K-1) \left( - u(n) + g(u(n) + \gamma_c K n)\right)}{ -\frac{\ol F}{f} \left(u(n) + \gamma_c K n \right)- ( -u(n) + g(u(n) + \gamma_c K n))} = \frac{1}{K} u'(n) \geq 0, & n \in (0, \infty).
\end{align*}
Hence, on $(-\infty,0)$, $u$ is an upper solution to the ODE \eqref{eq:rem:dealer:oli:norm:ODE-}, and on $( 0, \infty)$, it is a lower solution. Thus, on $(-\infty, 0)$, we can replace the upper solution $w$ in the proof of Theorem \ref{thm:ODE} by the smaller and whence tighter upper solution $u$ and conclude that $P'_* \leq u$ on $(-\infty, 0)$. In particular, we have $P_*'(0-) \leq u(0-)$. A similar argument on $(0, \infty)$ gives $P'_* \geq u$ on $(0,\infty)$ and $P_*'(0+) \geq u(0+)$. Together with \eqref{eq:rem:dealer:oli:norm:est:equil}, this establishes \eqref{eq:ass:ex:f:f-}--\eqref{eq:ass:ex:f:f+}.
\end{proof}

\section{Auxiliary Calculus Results}

For lack of easy references, this appendix collects a number of calculus results that are used in the proofs.

\begin{lemma}
\label{lem:increasing}
Let $-\infty \leq a < b \leq +\infty$ and $f: (a, b) \to \RR$. Suppose that each $x \in (a, b)$ has an open neighbourhood $U_x \subset (a, b)$ such that, for all $y \in U_x$,
	\begin{equation}
	f(y) 
	\begin{cases} < f(x) &\text{if } y < x, \\
 > f(x) \quad &\text{if } y > x.
 \end{cases}
 \label{eq:lem:increasing}
	\end{equation}
	Then $f$ is  increasing on $(a, b)$.
\end{lemma}

\begin{proof}
Seeking a contradiction, suppose there are $x_1, x_2 \in (a, b)$ with $x_1 < x_2$ and $f(x_1) \geq f(x_2)$. Set $I_{x_2} := \{x \leq x_2: f(x) > f(x_2)\}$ and $\tilde x_1 := \sup I_{x_2}$. Let $U_{\tilde x_1}$ be an open neighbourhood of $\tilde x_1$ such that \eqref{eq:lem:increasing} is satisfied. By the definition of $\tilde x_1$, there is $y \in U_{\tilde x_1} \cap I_{x_2}$ with $y < \tilde x_1$ such that $f(x_2) < f(y) < f(\tilde x_1)$. Hence $\tilde x_1 \in I_{x_2}$. It follows from \eqref{eq:lem:increasing} that $\tilde x_1 = x_2$. Let $U_{x_2}$ be such that \eqref{eq:lem:increasing} is satisfied for $x_2$. Then by definition of $\tilde x_1$, there is $y \in U_{x_2} \cap I_{x_2}$ such that $f(y) < f(x_2)$. This yields the desired contradiction and therefore shows $f$ is indeed increasing as asserted.
\end{proof}

\begin{lemma}
\label{lem:int by parts}
Let $\ol F, G: [0, \infty) \to \RR$ be absolutely continuous functions. Suppose that $\ol F(x) = \int_x^\infty f(y) \dd y$ for some nonnegative Borel function $f$ and $G(x) = \int_0^x g(y) \dd y$ for some locally integrable Borel function $g$.  Moreover, suppose there exists a nonnegative and nondecreasing function $H: [0, \infty) \to [0, \infty)$ with $|G| \leq H$ such that $\int_0^\infty H(x) f(x) \dd x < \infty$. Then
\begin{equation*}
\int_0^\infty G(x) f(x) dx = -\int_0^\infty g(x) \ol F(x) dx.
\end{equation*}
\end{lemma}

\begin{proof}
We may assume without loss of generality that $G$ is nondecreasing. Indeed, otherwise write $G = G^\uparrow - G^\downarrow$, where $G^\uparrow(x) = \int_0^x g^+(y) \dd y$ and $G^\downarrow(x) = \int_0^x g^-(y) \dd y$, and use linearity of the integral.

Fix $y > 0$. Integration by parts gives
\begin{equation*}
\int_0^y G(x) f(x) dx = G(y) \ol F(y) -\int_0^y g(x) \ol F(x) dx.
\end{equation*}
Moreover, by the assumptions on $H$ it follows that 
$$G(y) \ol F(y) \leq H(y) \ol F(y) \leq \int_y^{\infty} H(x) f(x) dx$$
By the assumption on $H$, we may conclude that $\lim_{y \to \infty} = G(y) \ol F(y)$. Now the claim follows from monotone convergence.
\end{proof}

\begin{lemma}
\label{lem:nowhere dense}
Let $-\infty \leq a <  b \leq \infty$ and $f: (a, b) \to \RR$ be continuously differentiable. Then:
\begin{enumerate}
\item $f$ is increasing if and only if $\{f' < 0\} = \emptyset$ and $\{f' = 0\}$ is nowhere dense.
\item For each compact set $K \subset \{f' > 0\}$ and any continuously differentiable function $g: (a, b) \to \RR$ that is supported on $K$, there is $\epsilon' > 0$ such that $f + \epsilon g$ is increasing for all $\epsilon \in [-\epsilon', \epsilon']$.
\end{enumerate}
\end{lemma}

\begin{proof}
(a) Note that the set $\{f' < 0\}$ is open and $\{f' = 0\}$ is closed in $(a, b)$ as $f'$ is continuous. 
``$\Rightarrow$'': If $f$ is increasing it is in particular nondecreasing and hence $\{f' < 0\} = \emptyset$. Seeking a contradiction, suppose that $\{f' = 0\}$ is not nowhere dense. Then there is an nonempty open set $U \subset \{f' = 0\}$. Then there is $a < c < d<b$ such that $[c, d] \in U$. It follows from the fundamental theorem of calculus that $f(c) = f(d)$, and we arrive at a contradiction.

``$\Leftarrow$'': As $\{f' < 0\} = \emptyset$, it follows that $f$ is nondecreasing. Seeking a contradiction, suppose there is $a < c < d<b$ such that $f(c) = f(d)$. As $f$ is nondecreasing, this implies that $f$ is constant on $(c, d)$ and hence $(c, d) \in \{f' = 0\}$, whence  $\{f' = 0\}$ fails to be nowhere dense and we arrive at a contradiction.

(b) Fix a compact set $K \in \{f' > 0\}$ and any continuously differentiable function $g: (a, b) \to \RR$ that is supported on $K$. Set $c_1 := \inf_{x \in K} f'(x) > 0$ and $c_2 := \sup_{x \in K}|g'(x)|$. By compactness of $K$, continuity of $f'$ and $g'$ and the fact that $K \in \{f' > 0\}$, it follows that $c_1 > 0$ and $c_2 < \infty$. Set $\epsilon' := \frac{c_1}{|c_2| +1}$. Then if $\epsilon \in [-\epsilon', \epsilon']$,
\begin{equation*}
f'(x) + \epsilon g'(x) 
\begin{cases}
\geq c_1 - \frac{c_1}{|c_2| +1} c_2 > 0&\text{if  }x \in K, \\
= f'(x) &\text{if  }x \in (a, b) \setminus K.
\end{cases}
\end{equation*}
It follows that $\{f' + \epsilon g' < 0\} = \{f' < 0\} = \emptyset$ and $\{f' + \epsilon g' = 0\} = \{f' = 0\}$. Hence $f + \epsilon g$ is strictly increasing by part (a).
\end{proof}	
	
\begin{proposition}\label{prop:walter}
Let $I = [0, \infty)$ or $(-\infty, 0]$. Let $v, w: I \to \RR$ be differentiable functions with $v \leq w$ and set $\Gamma:= \{(x,y) \in I \times \RR: v(x) \leq y \leq w(x)\}$. Finally let $\alpha: \Gamma \to \RR$ be a continuous function such that the partial derivative $\alpha_y: \Gamma \to \RR$ is also continuous (up to the boundary). Then the differential equation
\begin{equation}
\label{eq:prop:walter:ODE}
y '(x)= \alpha(x,y(x))
\end{equation}
has global solution $\phi$ on $I$ that satisfies $v \leq \phi \leq w$ if either $I = [0, \infty)$ and  
\begin{equation}
\label{eq:prop:walter:R+}
w'(x) - \alpha(x,w(x)) \leq 0 \leq v'(x) - \alpha(x,v(x)),
\end{equation}
or $I = (-\infty, 0]$ and
\begin{equation}
\label{eq:prop:walter:R-}
v'(x) - \alpha(x,v(x)) \leq 0 \leq w'(x) - \alpha(x,w(x)).
\end{equation}
We call $v$ a \emph{lower} and $w$ an \emph{upper} solution to \eqref{eq:prop:walter:ODE}
\end{proposition}

\begin{proof}
Because $\Gamma$ is closed and $\alpha$ and $\alpha_y$ are continuous on $\Gamma$, we can extend $f$ to a continuous function $\bar \alpha: I \times \RR \to\RR$ with continuous partial derivate $\bar \alpha_y$ such that $\bar \alpha$ and $\alpha$ as well as $\bar \alpha_y$ and $\alpha_y$ coincide on $\Gamma$. Now \eqref{eq:prop:walter:R+} follows from \cite[Theorem \S 9 XIII]{walter:96}. Finally,  \eqref{eq:prop:walter:R-} can be reduced to \eqref{eq:prop:walter:R+} by setting $\tilde v(x) = v(-x)$, $\tilde w(x) = w(-x)$ and $-\tilde \alpha(x, y) = \alpha(-x, y)$.
	\end{proof}


\section{Log-Concave Distributions}

In this appendix, we list some well-known and not so well-known facts about log-concave distributions; see \cite{an.98} and \cite{saumard:wellner:14} for general overviews on log-concave distributions.

First, we recall some basic properties of log-concave distributions on the real line. 
\begin{proposition}
	\label{prop:log:concave:basics}
Let $f: \RR \to (0, \infty)$ be a log-concave probability density function. Denote by $F(x) = \int_{-\infty}^x f(y) \dd y$ and $\ol F(x) = \int_{-\infty}^x f(y) \dd y$, $x \in \RR$ its cumulative distribution function and survival function, respectively. Then:
\begin{enumerate}
\item both $F$ and $\ol F$ are log-concave as well;
\item there exist $C > 0$ and $\epsilon > 0$ such that $f(x) \leq C \exp(-\epsilon |x|)$ for all $x \in \RR$.
\item $f$ admits a right derivative $f'$ everywhere and there exists  $x^*  \in \RR$ such that $f' \geq 0$  on $(-\infty, x^*)$ and $f' \leq 0$ on $ [x^*, \infty)$.
\end{enumerate}
\end{proposition}

\begin{proof}
Part (a) follows from \cite[Lemma 3]{an.98}. Part (b) is a consequence of \cite[Corollary~1(ii)]{an.98} and the fact that $x \mapsto f(-x)$ is also log-concave. Existence of a right-derivative $f'$ follows from the fact that $\log(f)$ admits a right derivative everywhere since it is concave. Finally, the existence of $x^*$ is implied by the fact that $f$ is (strongly) unimodal by \cite[Proposition~1]{an.98}. 
\end{proof}

Next, we show that convolutions preserve log-concavity and yield additional regularity.\footnote{Note that \emph{both} $f$ and $g$ need to be log-concave: \cite[Proposition 16]{biais.al.00} is false; see \cite{miravete:02} for a counterexample.}
\begin{proposition}
	\label{prop:log concave:convolution}
Let $f, g: \RR \to (0, \infty)$ be log-concave probability density functions and let $f'$ denote the right derivative of $f$. Moreover, denote by $\id$ the identity function.
\begin{enumerate}
\item The convolution $f * g$ is again a log-concave probability density function, and continuously differentiable with bounded derivative $(f * g)' = f' * g$;
\item The convolution $f * (\id\, g)$ is integrable and continuously differentiable with bounded derivative $(f* (\id \,g))' = f' * (\id \,g)$.
\end{enumerate}

\begin{proof}
The first part of (a) follows from \cite[Proposition 4]{an.98}. 
	
For the remainder of (a) and (b), fix $x^*  \in \RR$ as in Proposition \ref{prop:log:concave:basics}(c). The fundamental theorem of calculus yields
\begin{align}
	\int_{-\infty}^\infty |f'(x)|\dd x &= \int_{-\infty}^{x^*} f'(x) \dd x - \int_{-\infty}^{x^*} f'(x) \dd x  = 2 f(x^*) < \infty,
\label{eq:pf:prop:log concave:convolution:f}
\end{align}
Since $g$ and $\id \, g$ are bounded Proposition \ref{prop:log:concave:basics}(b), the convolutions $f' * g$ and $f' * (\id \, g)$ are well-defined, continuous (by dominated convergence) and bounded. Now the result follows from the fundamental theorem of calculus and Fubini's theorem.
\end{proof}
\end{proposition}

Finally, we derive a refined version of Efron's theorem~\citep{efron:65} on the conditional mean of a log-concave random variable given the sum of this random variable and another independent log-concave random variable.
\begin{proposition}
		\label{prop:cond:mean}
		Let $U$ and $V$ be independent real-valued random variables with positive log-concave probability density functions $f_U$ and $f_V$. Set $W = U + V$. Then the conditional mean function 
		$$w \mapsto g(w) = E[U\,|\, W = w]$$ 
		 is continuously differentiable and satisfies $g' > 0$.
\end{proposition}

\begin{proof}
First, $g$ is continuously differentiable since
\begin{equation*}
g(w) = \frac{\int_\RR (w- v)f_U(w -v) f_V(v) \dd v}{\int_\RR f_U(w-v) f_V(v) \dd v}
	\end{equation*}
and both the numerator and denominator are continuously differentiable by Proposition \ref{prop:log concave:convolution}, with  derivatives $((\id\, f_U) * f_V )' = (\id \,f_U) * f_V ' $ and $(f_U * f_V )' = f_U * f_V '$ respectively, where $f'_V$ denotes the right derivative of $f_V$.

To show that $g' > 0$, fix $w \in \RR$. Then Fubini's theorem gives
\begin{align}	
\label{eq:pf:prop:cond:mean:g'}
g'(w) &= \frac{\int_\RR \int_\RR (w-v_1) f_U(w-v_1) f_U(w-v_2) \big[ f'_V(v_1) f_V(v_2) -   f'_V(v_2) f_V(v_1)\big]\dd v_1 \dd v_2}{\int_\RR \int_\RR  f_U(w-v_1) f_U(w-v_2) f_V(v_1)f_V(v_2) \dd v_1 \dd v_2}.
\end{align}
To complete the proof, it remains to show that the numerator in \eqref{eq:pf:prop:cond:mean:g'} is positive. Using symmetry and  averaging over the first and second line for the third line, we obtain
\begin{align}
&\int_\RR \int_\RR (w-v_1) f_U(w-v_1) f_U(w-v_2) \big[ f'_V(v_1) f_V(v_2) -   f'_V(v_2) f_V(v_2)\big]\dd v_1 \dd v_2 \notag \\
&=\int_\RR \int_\RR (v_2-w) f_U(w-v_1) f_U(w-v_2) \big[ f'_V(v_1) f_V(v_2) -   f'_V(v_2) f_V(v_2)\big]\dd v_1 \dd v_2 \notag  \\
&= \frac{1}{2} 	\int_\RR \int_\RR (v_2-v_1) f_U(w-v_1) f_U(w-v_2) \big[ f'_V(v_2) f_V(v_1) -   f'_V(v_1) f_V(v_2)\big]\dd v_1 \dd v_2 \notag \\
&=\frac{1}{2} 	\int_\RR \int_\RR \left[(v_2-v_1)\left(\frac{f'_V(v_1)}{f_V(v_1)} - \frac{f'_V(v_2)}{f_V(v_2)}  \right)\right] f_U(w-v_1) f_U(w-v_2)  f_V(v_1) f_V(v_2) \dd v_1 \dd v_2. \label{eq:eq:pf:prop:cond:mean:g':num}
\end{align}
By log-concavity of $f$, the function $v \mapsto f'_V(v)/f_V(v)$ (which is the right derivative of $\log (f_V)$) is nonincreasing. This implies that $f'_V(v_1)/f_V(v_1) - f'_V(v_2)/f_V(v_2)$ is nonpositive for $v_1 \leq v_2$ and nonnegative for $v_1 \geq v_2$. Moreover, since $f$ is integrable, $v \mapsto f'_V(v)/f_V(v)$ is not constant. As a consequence,
\begin{equation*}
(v_2-v_1)\left(\frac{f'_V(v_1)}{f_V(v_1)} - \frac{f'_V(v_2)}{f_V(v_2)}  \right) \geq 0,
\end{equation*}
where for each $v_1 \in \RR$ sufficiently small, the inequality is strict if $|v_2 - v_1|$ is sufficiently large (because $v \mapsto f'_V(v)/f_V(v)$ is nonincreasing and not constant). Since $f_U(w-v_1) f_U(w-v_2)  f_V(v_1) f_V(v_2)$ is positive for each $v_1, v_2 \in \RR$, it follows that \eqref{eq:eq:pf:prop:cond:mean:g':num} is positive.
\end{proof}

\bibliographystyle{abbrvnat}
\bibliography{References}

\begin{thebibliography}{15}
\providecommand{\natexlab}[1]{#1}
\providecommand{\url}[1]{\texttt{#1}}
\expandafter\ifx\csname urlstyle\endcsname\relax
  \providecommand{\doi}[1]{doi: #1}\else
  \providecommand{\doi}{doi: \begingroup \urlstyle{rm}\Url}\fi

\bibitem[An(1998)]{an.98}
M.~Y. An.
\newblock Logconcavity versus logconvexity: a complete characterization.
\newblock \emph{Journal of Economic Theory}, 80\penalty0 (2):\penalty0
  350--369, 1998.

\bibitem[Attar et~al.(2019)Attar, Mariotti, and Salani{\'e}]{attar.al.19}
A.~Attar, T.~Mariotti, and F.~Salani{\'e}.
\newblock On competitive nonlinear pricing.
\newblock \emph{Theoretical Economics}, 14\penalty0 (1):\penalty0 297--343,
  2019.

\bibitem[Back and Baruch(2013)]{back.baruch.13}
K.~Back and S.~Baruch.
\newblock Strategic liquidity provision in limit order markets.
\newblock \emph{Econometrica}, 81\penalty0 (1):\penalty0 363--392, 2013.

\bibitem[Bertsekas(1999)]{bertsekas.99}
D.~Bertsekas.
\newblock \emph{Nonlinear programming}.
\newblock Athena Scientific, Belmont, MA, second edition, 1999.

\bibitem[Biais et~al.(2000)Biais, Martimort, and Rochet]{biais.al.00}
B.~Biais, D.~Martimort, and J.-C. Rochet.
\newblock Competing mechanisms in a common value environment.
\newblock \emph{Econometrica}, 68\penalty0 (4):\penalty0 799--837, 2000.

\bibitem[Biais et~al.(2013)Biais, Martimort, and Rochet]{biais.al.13}
B.~Biais, D.~Martimort, and J.-C. Rochet.
\newblock Corrigendum to {``Competing mechanisms in a common value
  environment''}.
\newblock \emph{Econometrica}, 81\penalty0 (1):\penalty0 393--406, 2013.

\bibitem[Bielagk et~al.(2019)Bielagk, Horst, and
  Moreno-Bromberg]{bielagk.al.19}
J.~Bielagk, U.~Horst, and S.~Moreno-Bromberg.
\newblock Trading under market impact: Crossing networks interacting with
  dealer markets.
\newblock \emph{Journal of Economic Dynamics and Control}, 100:\penalty0
  131--151, 2019.

\bibitem[Cetin and Waelbroeck(2021)]{cetin.21}
U.~Cetin and H.~Waelbroeck.
\newblock An equilibrium analysis of price impact and order flow.
\newblock Preprint, 2021.

\bibitem[Efron(1965)]{efron:65}
B.~Efron.
\newblock Increasing properties of {Polya} frequency functions.
\newblock \emph{Annals of Mathematical Statistics}, 36\penalty0 (1):\penalty0
  272--279, 1965.

\bibitem[Glosten(1989)]{glosten.89}
L.~R. Glosten.
\newblock Insider trading, liquidity, and the role of the monopolist
  specialist.
\newblock \emph{Journal of Business}, 62\penalty0 (2):\penalty0 211--235, 1989.

\bibitem[Ho and Stoll(1981)]{ho.stoll.81}
T.~Ho and H.~R. Stoll.
\newblock Optimal dealer pricing under transactions and return uncertainty.
\newblock \emph{Journal of Financial Economics}, 9\penalty0 (1):\penalty0
  47--73, 1981.

\bibitem[Miravete(2002)]{miravete:02}
E.~J. Miravete.
\newblock Preserving log-concavity under convolution: Comment.
\newblock \emph{Econometrica}, 70\penalty0 (3):\penalty0 1253--1254, 2002.

\bibitem[Saumard and Wellner(2014)]{saumard:wellner:14}
A.~Saumard and J.~A. Wellner.
\newblock Log-concavity and strong log-concavity: A review.
\newblock \emph{Statistics Surveys}, 8:\penalty0 45, 2014.

\bibitem[Treynor(1971)]{treynor.71}
J.~Treynor.
\newblock The only game in town.
\newblock \emph{Financial Analysts Journal}, 22:\penalty0 12--14, 1971.

\bibitem[Walter(1998)]{walter:96}
W.~Walter.
\newblock \emph{Ordinary differential equations}.
\newblock Springer, New York, 1998.

\end{thebibliography}
\end{document}